\documentclass[a4paper,USenglish,cleveref,autoref,thm-restate,nolineno]{lipics-v2021}

\hideLIPIcs

\newtheorem{longtheorem}{Theorem}
\newtheorem{longlemma}[longtheorem]{Lemma}

\title{Extending Orthogonal Planar Graph Drawings is Fixed-Parameter Tractable } 

\titlerunning{Extending Orthogonal Planar Graph Drawings is FPT} 

\usepackage{complexity}
\usepackage{cite,xspace}
\usepackage[textsize=footnotesize]{todonotes}
\Crefname{property}{Property}{Properties}
\usepackage{tikz, float}
\usetikzlibrary{positioning,backgrounds,patterns,decorations.pathreplacing,calligraphy}
\DeclareMathOperator{\bdist}{bd}
\DeclareMathOperator{\bvect}{vect}
\tikzset{
  circ/.style = {circle,draw,fill=white,inner sep=1.3pt}
}
\tikzset{
 wcirc/.style = {circle,draw,fill,inner sep=3.5pt,color=white}
}
\tikzset{
 bcirc/.style = {circle,draw,fill,inner sep=1pt,color=black}
}
\def\calG{\mathcal G}
\def\calF{\mathcal F}
\def\calU{\mathcal U}
\def\calO{\mathcal O}
\def\calT{\mathcal T}
\def\calP{\mathcal P}

\newcommand{\bigoh}{\mathcal{O}}

\DeclareMathOperator{\tw}{tw}
\newcommand{\gridsols}{\texttt{gridsols}}

\newcommand{\old}[1]{{}}

\newcommand{\TBOE}{\textsc{BMOE}\xspace}
\newcommand{\FTBOE}{\textsc{F-BMOE}\xspace}
\newcommand{\directions}{\{\uparrow,\downarrow,\leftarrow,\rightarrow\}}

\newcommand{\sequiv}{shape-equivalent\xspace}
\newtheorem{property}[lemma]{Property}

 \newcommand{\rec}{\texttt{Record}}

 \author{Sujoy Bhore}{Indian Institute of Technology Bombay, India}{sujoy.bhore@gmail.com}{https://orcid.org/0000-0003-0104-1659}{}

 \author{Robert Ganian}{Technische Universit\"at Wien, Austria}{rganian@ac.tuwien.ac.at}{https://orcid.org/0000-0002-7762-8045}{}

 \author{Liana Khazaliya}{Technische Universit\"at Wien, Austria}{liana.khazaliya@tuwien.ac.at}{}{}

 \author{Fabrizio Montecchiani}{Università degli Studi di Perugia, Italy}{fabrizio.montecchiani@unipg.it}{https://orcid.org/0000-0002-0543-8912}{}

 \author{Martin N\"ollenburg}{Technische Universit\"at Wien, Austria}{noellenburg@tuwien.ac.at}{[orcid]}{}

 \authorrunning{S. Bhore, R. Ganian, L. Khazaliya, F. Montecchiani, M. N\"ollenburg} 
\authorrunning{Anonymous Author(s)} 

\Copyright{Sujoy Bhore, Robert Ganian, Liana Khazaliya, Fabrizio Montecchiani, Martin N\"ollenburg} 

\ccsdesc[500]{Theory of computation~Computational geometry}
\ccsdesc[500]{Mathematics of computing~Graph algorithms}

\keywords{orthogonal drawings, bend minimization, extension problems, parameterized complexity} 

\nolinenumbers 

\begin{document}

\maketitle

\begin{abstract}
The task of finding an extension to a given partial drawing of a graph while adhering to constraints on the representation has been extensively studied in the literature, with well-known results providing efficient algorithms for fundamental representations such as planar and beyond-planar topological drawings. In this paper, we consider the extension problem for bend-minimal orthogonal drawings of planar connected graphs, which is among the most fundamental geometric graph drawing representations. While the problem was known to be \NP-hard, it is natural to consider the case where only a small part of the graph is still to be drawn. Here, we establish the fixed-parameter tractability of the problem when parameterized by the size of the missing subgraph.
Our algorithm is based on multiple novel ingredients which intertwine geometric and combinatorial arguments. These include the identification of a new graph representation of bend-equivalent regions for vertex placement in the plane, establishing a bound on the treewidth of this auxiliary graph, and a global point-grid that allows us to discretize the possible placement of bends and vertices into locally bounded subgrids for each of the above regions. 
\end{abstract}

\section{Introduction}
\label{sec:intro}

Extending partial drawings of graphs while preserving certain desirable properties such as planarity is an algorithmic problem that received considerable attention in the last decade in graph theory, graph drawing, and computational geometry. 
Drawing extension problems are motivated, for instance, by visualizing networks, in which certain subgraphs represent important motifs that require a specific drawing, or by visualizing dynamic networks, in which new edges and vertices must be integrated in an existing, stable drawing. 
Generally speaking, we are given a graph $G$ and a (typically connected) subgraph $H$ of $G$ with a drawing $\Gamma(H)$, which is called a \emph{partial} drawing of $G$. 
The drawing $\Gamma(H)$ typically satisfies certain topological or geometric properties, e.g., planarity, upward planarity, or 1-planarity, and the goal of the corresponding extension problem is to extend $\Gamma(H)$ to a drawing $\Gamma(G)$ of the whole graph $G$ (if possible) by inserting the missing vertices and edges into $\Gamma(H)$ while maintaining the required drawing properties. 

A fundamental result in this line of research is the work of Angelini et al.~\cite{adfjkp-tppeg-15}, who showed that for planar graphs with a given partial planar drawing, the extension problem can be solved in linear time, thus matching the time complexity of unconstrained planarity testing.
In fact, there is also a corresponding combinatorial characterization of planar graphs with extensible partial planar drawings via forbidden substructures~\cite{jkr-ktppeg-13}.
In contrast to the above results, which consider topological graph embeddings, the planar drawing extension problem is \NP-hard in its geometric variant, where one has to decide if a partial planar straight-line drawing $\Gamma(H)$ can be extended to a planar straight-line drawing of $G$~\cite{Patrignani06}.

In this paper, we study the geometric drawing extension problem arising in the context of one of the most fundamental graph drawing styles: orthogonal drawings~\cite{efk-ogd-01,dg-popda-13,DBLP:books/ph/BattistaETT99,DBLP:books/ws/NishizekiR04}. 
In a planar orthogonal drawing, edges are represented as 
polylines comprised of (one or more) horizontal and vertical segments with as few overall bends as possible, where edges are not allowed to intersect except at common endpoints.
Orthogonal drawings find applications in various domains from VLSI and printed circuit board (PCB) design, to schematic network visualizations, e.g., UML diagrams in software engineering, argument maps, or flow charts. 

Given the above, a key optimization goal in orthogonal drawings is bend minimization. This task is known to be \NP-hard~\cite{GargT01} when optimizing over all possible combinatorial embeddings of a given graph, but can be solved in polynomial time for a fixed combinatorial embedding using the network flow model of Tamassia~\cite{t-eggmn-87}. 
Interestingly, the complexity of the bend minimization problem without a fixed embedding depends on the vertex degrees, which in the classical case of vertices being represented as points is naturally bounded by $4$.
If, however, the maximum vertex degree is $3$, then there is a polynomial-time algorithm for bend minimization~\cite{BattistaLV98}, and this result has recently been improved to linear time~\cite{dlop-oodp3lt-20}; more generally, the problem is fixed-parameter tractable (FPT) in the number of degree-$4$ vertices~\cite{dl-codves-98}. In addition, it has been recently shown that the bend minimization problem is in \XP\ (slice-wise polynomial) parameterized by the treewidth of the input graph~\cite{DBLP:journals/jcss/GiacomoLM22}.


Despite the general popularity of planar orthogonal graph drawings, the corresponding extension problem has only been considered recently~\cite{AngeliniRS21}. 
While the authors of that paper showed that the existence of a planar orthogonal extension can be decided in linear time, the orthogonal bend-minimal drawing extension problem in general is easily seen to be \NP-complete as it generalizes the case in which the pre-drawn part of the graph is empty~\cite{GargT01}. 
Our paper addresses the parameterized complexity of the bend-minimal extension problem for planar orthogonal graph drawings under the most natural parameterization of the problem, which is the size of the subgraph that is still missing from the drawing. This parameter can be assumed to be small in many applications, e.g., when extending drawings of dynamic graphs with few added edges and vertices, and has been used broadly in the study of previous topological drawing extension problems. 

\subparagraph{Contributions.}
In this paper, we establish the fixed-parameter tractability of the \textsc{Bend-Minimal Orthogonal Extension (\TBOE)} problem when parameterized by the size $\kappa$ of the missing subgraph (see the formal problem statement in \cref{se:preliminaries}). A general difficulty we had to overcome on our way to obtain our fixed-parameter algorithm is the fact that while there have been numerous recent advances in the parameterized study of drawing extension problems~\cite{eghkn-ep1d-20,GanianHKPV21,HammH22}, 
the specific drawing styles considered in those papers were primarily topological in nature, while for bend minimization the geometry of the instance is crucial. In order to overcome this difficulty, we develop a new set of tools summarized below.

In Section~\ref{sec:red}, we make the first and simplest step towards fixed-parameter tractability of \TBOE by applying an initial branching step to simplify the problem. This step allows us to reduce our target problem to \textsc{Bend-Minimal Orthogonal Extension on a Face (\FTBOE)}, where the missing edges and vertices are drawn only in a marked face $f$ and we have some additional information about how the edges are geometrically connected. 

Next, in Section~\ref{sec:prep}, we focus on solving an instance of \FTBOE. We show that certain parts of the marked face $f$ are irrelevant and can be pruned away, and also use an involved argument to reduce the case of $f$ being the outer face to the case of $f$ being an inner face. 

Once that is done, we enter the centerpiece of our approach in Section~\ref{sec:discret}, where the aim is to obtain a suitable discretization of our instance. To this end, we split the face $f$ into so-called \emph{sectors}, which group together points that have the same ``bend distances'' to all of the connecting points on the boundary of $f$. Furthermore, we construct a \emph{sector-grid}---a point set such that each sector contains a bounded number of points from this set, and every bend-minimal extension can be modified to only use points from this set for all vertices and bends. While this latter result would make it easy to handle each individual sector by brute force, the issue is that the number of sectors can be very large, hindering tractability. 

To deal with this obstacle, we capture the connections between sectors via a \emph{sector graph} whose vertices are precisely the sectors and edges represent geometric adjacencies between sectors. 
Crucially, in Section~\ref{sec:tw} we show that the sector graph has treewidth bounded by a function of $\kappa$. This is non-trivial and relies on the previous application of the pruning step in Section~\ref{sec:prep}. Having obtained this bound on the treewidth, the last step simply combines the already constructed sector grid with dynamic programming to solve \FTBOE (and hence also \TBOE). It is perhaps worth pointing out the interesting contrast between the use of treewidth here as an implicit structural property of the sector graph---a crucial tool in our fixed-parameter algorithm---with the previously considered use of treewidth directly on the input graph---which is not known to lead to fixed-parameter tractability~\cite{DBLP:journals/jcss/GiacomoLM22}.

\subparagraph{Related work.}
Several variants of drawing extension problems have been studied over the years. For instance, Chambers et al.~\cite{cegl-dgpwpofpa-12} studied the problem of drawing a planar graph using straight-line edges with a prescribed convex polygon as the outer face, and proposed a method that produces drawings with polynomial area. 
Mchedlidze et al.~\cite{mnr-ecpdg-15} provide a characterization (which can be tested in linear time) to determine whether given a planar straight-line convex drawing of a biconnected subgraph $G'$ of a planar graph $G$ with a fixed planar embedding, this drawing can be extended to a planar straight-line drawing of $G$. 
Recently, Eiben et al. studied the problem of extending 1-planar drawings. While the problem was known to be \NP-complete, they showed~\cite{eghkn-ep1d-20} that the problem is \FPT\ when parameterized by the edge deletion distance. 
Later, in~\cite{eghkn-enc1dpt-20}, they showed that the 1-planar extension is polynomial-time solvable when the number of vertices and edges to be added to the partial drawing is bounded.
Hamm and Hlin\v en\' y also studied the parameterized complexity of the extension problem in the setting of crossing minimization~\cite{HammH22}.

Other types of extension problems have also been investigated, e.g., Da Lozzo et al.~\cite{lbf-eupgd-20} studied the upward planarity extension problem, and showed that this is \NP-complete even for very restricted settings. 
Br\"{u}ckner and Rutter~\cite{br-pclp-17} showed that the partial level planarity problem is \NP-complete again in severely restricted settings. 
For non-planar graph drawings, it is even \NP-hard to determine whether a single edge can be inserted into a simple partial drawing of the remaining graph, i.e., a drawing in which any two edges intersect in at most one point~\cite{arroyo2022inserting}.
Extension problems have been investigated also for other types of graph representations, in particular for intersection representations such as circular arc graphs~\cite{FialaRSZ22} or circle graphs~\cite{BrucknerRS22}.
In the context of bend-minimal planar orthogonal drawing extension, Angelini et al. showed that the problem remains \NP-hard even when a planar embedding of the whole graph is provided in the input~\cite{AngeliniRS21}.

\section{Preliminaries and Basic Tools}\label{se:preliminaries}

We assume familiarity with basic concepts in parameterized complexity theory, notably fixed-parameter tractability~\cite{CyganFKLMPPS15}. 

\subparagraph{Treewidth.}
A \emph{tree decomposition}~$\mathcal{T}_G$ of a graph $G=(V,E)$ is a pair 
$(T,\chi)$, where $T$ is a tree (whose vertices we call \emph{nodes}) rooted at a node $r$ and $\chi$ is a function that assigns each node $t$ a set $\chi(t) \subseteq V$ such that the following holds: 
for every $uv \in E$ there is a node	$t$ such that $u,v\in \chi(t)$, and for every vertex $v \in V$, the set of nodes $t$ satisfying $v\in \chi(t)$ forms a nonempty subtree of~$T$. 

\begin{itemize}
	\item For every $uv \in E$ there is a node	$t$ such that $u,v\in \chi(t)$.
	\item For every vertex $v \in V$, the set of nodes $t$ satisfying $v\in \chi(t)$ forms a nonempty subtree of~$T$.
\end{itemize}

A tree decomposition is \emph{nice} if the following two conditions are also satisfied:
\begin{itemize}
	\item $|\chi(\ell)|=0$ for every leaf $\ell$ of $T$ and $|\chi(r)|=0$.
	\item There are only three kinds of non-leaf nodes in $T$:
	\begin{itemize}
        \item \textbf{Introduce node:} a node $t$ with exactly
          one child $t'$ such that $\chi(t)=\chi(t')\cup
          \{v\}$ for some vertex $v\not\in \chi(t')$.
        \item \textbf{Forget node:} a node $t$ with exactly
          one child $t'$ such that $\chi(t)=\chi(t')\setminus
          \{v\}$ for some vertex $v\in \chi(t')$.
        \item \textbf{Join node:} a node $t$ with two children $t_1$,
          $t_2$ such that $\chi(t)=\chi(t_1)=\chi(t_2)$.
	\end{itemize}
\end{itemize}
The \emph{width} of a tree decomposition $(T,\chi)$ is the size of a largest set $\chi(t)$ minus~$1$, and the \emph{treewidth} of the graph $G$
is the minimum width of a tree decomposition of~$G$. We use $T_t$ to denote the subtree of $T$ rooted at $t$, and $\chi_\downarrow(t)$ to denote the set $\bigcup_{t' \in V(T_t)}\chi(t')$.

\subparagraph{Basic definitions.} 
Let $G$ be a simple connected  planar graph with vertex degree at most four. We use the notation $V(G)$ and $E(G)$ to denote the vertex and edge set of $G$. A \emph{drawing} $\Gamma(G)$ is a mapping of the vertices in $V(G)$ to points of the plane, and of the edges in $E(G)$ to Jordan arcs connecting their corresponding endpoints but not passing through any other vertex.
We only consider \emph{simple drawings}, in which any two arcs representing two edges have at most one point in common, which is either a common endpoint or a common interior point where the two arcs properly cross each other. 
A drawing is \emph{planar} if no two edges cross each other. 
It follows from the above definitions that in a simple planar drawing any two edges share at most one point which is a common endpoint. 
A planar drawing partitions the plane into topologically connected regions called \emph{faces}, one of which is unbounded and called the \emph{outer face}, in contrast with all other faces which are \emph{inner faces}.
A planar drawing $\Gamma(G)$ is \emph{orthogonal} if each edge is a polyline consisting of horizontal and vertical segments.
A \emph{bend} in a polygonal chain representing an edge in $\Gamma(G)$ is a point shared by two consecutive segments of the chain. When this creates no ambiguities, we make no distinction between the vertices of $G$ and the corresponding points of $\Gamma(G)$, as well as between the edges of $G$ and the corresponding polylines of $\Gamma(G)$.
For instance, \cref{fi:example-1} shows an orthogonal drawing of a graph $G$, in which edge $ax$ has three bends.

\begin{figure}[t]
    \centering
    \begin{minipage}[t]{0.45\textwidth}
    \centering
        \includegraphics[page=2, scale=1]{figs/extexample.pdf}
        \subcaption{$\Gamma(G)$\label{fi:example-1}}
    \end{minipage}
    \begin{minipage}[t]{0.45\textwidth}
    \centering
        \includegraphics[page=1, scale=1]{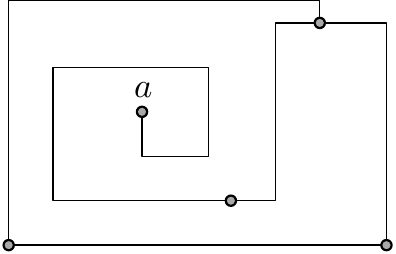}
        \subcaption{$\Gamma(H)$\label{fi:example-2}}
    \end{minipage}
    \caption{An orthogonal drawing of (a) a graph $G$ and (b) a subgraph $H$ of $G$.}
    \label{fig:example}
\end{figure}

\subparagraph{Problem Statement.}
Let $G$ be a planar graph and $H$ be a connected subgraph of $G$. 
We call the complement $X = V(G) \setminus V(H)$ the \emph{missing vertex set} of $G$, and $E_X = E(G) \setminus E(H)$ the \emph{missing edge set}. 
Let $\Gamma(H)$ be a planar orthogonal drawing of $H$. 
A planar orthogonal drawing $\Gamma(G)$ \emph{extends} $\Gamma(H)$ if its restriction to the vertices and edges of $H$ coincides with $\Gamma(H)$. 
Moreover, $\Gamma(G)$ is a $\beta$-extension of $\Gamma(H)$ if it extends $\Gamma(H)$ and the total number of bends along the edges of $E_X$ is at most $\beta$, for some $\beta \in \mathbb{N}$. For example, \cref{fi:example-1} shows a 7-extension $\Gamma(G)$ of the drawing $\Gamma(H)$ in \cref{fi:example-2}, with the missing vertices drawn in red.  

Our problem of interest is defined as follows. 

\medskip\noindent\fbox{%
  \parbox{0.95\textwidth}{
  \textsc{Bend-Minimal Orthogonal Extension (\TBOE)}\\
  \textbf{Input:} ($G, H, \Gamma(H)$), integer $\beta$\\
  \textbf{Problem:} Is there a $\beta$-extension $\Gamma(G)$ of $\Gamma(H)$?
  }}
\medskip

We remark that $\TBOE$ is known to be \NP-hard even when restricted to the case where $\beta=0$ and $V(H)=\emptyset$~\cite{GargT01}.   Also, unless specified otherwise, in the rest of the paper we only consider orthogonal drawings which are planar.
Our parameter of interest is the number of vertices and edges missing from $H$, i.e., $\kappa=|V(G)\setminus V(H)|+|E(G)\setminus E(H)|$.

\subparagraph{Basic Tools.}
We 
introduce a set of redrawing operations that will be used as basic tools in several 
proofs. It is worth noting that similar operations as the ones introduced here, which are based on shortening or prolonging sets of parallel edges in orthogonal drawings, are well known (see, e.g.,~\cite{BLPS13}). However, in our specific setting we have parts of the drawing that are given and cannot be modified, and handling this requires additional care in our arguments.

A \emph{feature point} of an orthogonal drawing is a point representing either a vertex or a bend of an edge. 
An \emph{edge-segment} of an orthogonal drawing is a segment that belongs to a polyline representing an edge. 
Two orthogonal drawings $\Gamma(G)$ and $\Gamma'(G)$ of a planar graph $G$ are \emph{\sequiv} if one can be obtained from the other by only shortening or lengthening some edge-segments. \cref{fig:operations} shows an example of two \sequiv drawings; in particular, the one on the right can be obtained from the one on the left by suitably shortening the blue (thicker) edge-segments. (We note that in the literature on orthogonal drawings, this is equivalent to saying that $\Gamma(G)$ and $\Gamma'(G)$ have the same \emph{shape} but two different \emph{metrics}.) 

\begin{figure}[t]
    \centering
    \begin{minipage}[t]{0.45\textwidth}
     \centering
        \includegraphics[page=1,width=\textwidth]{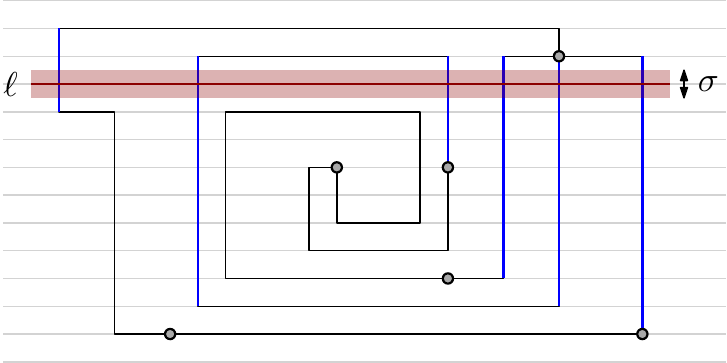}
    \end{minipage}\hfill
    \begin{minipage}[t]{0.45\textwidth}
     \centering
        \includegraphics[page=2,width=\textwidth]{figs/operations.pdf}
    \end{minipage}
    \caption{Two \sequiv orthogonal drawings such that the one on the right is obtained from the one on the left by applying a $(\sigma,\ell)$-strip removal operation.}
    \label{fig:operations}
\end{figure}

Let $\Gamma(G)$ be an orthogonal drawing of a graph $G$. Let $\ell$ be a horizontal (vertical) line that contains no feature points of $\Gamma(G)$ but intersects a set $S$ of vertical (horizontal) edge-segments. Let $l$ be the shortest distance between the endpoints of the segments in $S$ and $\ell$. For any $0 < \sigma < l$, a \emph{$(\sigma,\ell)$-strip removal operation} consists of decreasing the $y$-coordinates ($x$-coordinates) of all feature points above  (to the right of) $\ell$ by $\sigma$. Analogously, for any $\sigma > 0$, a \emph{$(\sigma,\ell)$-strip addition operation} consists of increasing the $y$-coordinates ($x$-coordinates) of all feature points above (to the right of) $\ell$ by $\sigma$. See \cref{fig:operations} for an illustration of a $(\sigma,\ell)$-strip removal operation. The following property readily follows.

\begin{property}\label{prop:sigma}
Let $\Gamma(G)$ and $\Gamma'(G)$ be two orthogonal drawings  such that $\Gamma'(G)$ is obtained from $\Gamma(G)$ by applying a $(\sigma,\ell)$-strip removal or addition operation. Then $\Gamma(G)$ and $\Gamma'(G)$ are \sequiv.
\end{property}

\begin{figure}[t]
    \centering
    \begin{minipage}[t]{0.45\textwidth}
     \centering
        \includegraphics[page=3,width=\textwidth]{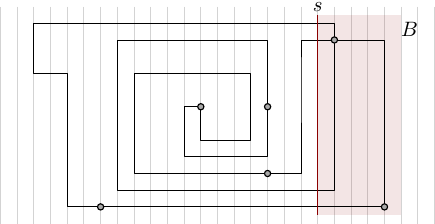}
    \end{minipage}\hfill
    \begin{minipage}[t]{0.45\textwidth}
     \centering
        \includegraphics[page=4,width=\textwidth]{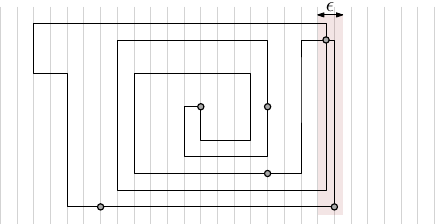}
    \end{minipage}
    \caption{Two \sequiv orthogonal drawings such that the one on the right is obtained from the one on the left by applying \cref{le:box} with the v-selection $B$.}
    \label{fig:operations-2}
\end{figure}

Let $B$ be a rectangle that intersects $\Gamma(G)$ such that only one side $s$ of $B$ is crossed by edges of $G$. We call $B$ a \emph{v-selection} if $s$ is vertical or a \emph{h-selection} otherwise. Also, the subdrawing of $\Gamma(G)$ inside $B$  is called the \emph{$B$-selected drawing}; see \cref{fig:operations-2} for an illustration of a v-selection and of the next lemma (whose proof easily follows from \cref{prop:sigma}). 

\begin{lemma}\label{le:box}
Let $\Gamma(G)$ be an orthogonal drawing and let $B$ be a $v$-selection ($h$-selection) of $\Gamma(G)$. For any $\epsilon>0$, there is a drawing $\Gamma'(G)$ that is \sequiv to $\Gamma(G)$ and such that the $B$-selected drawing has width (height) at most $\epsilon$ and height (width) equal as in $\Gamma(G)$.
\end{lemma}
\begin{proof}
We prove the case in which $B$ is a v-selection, as the case in which it is a h-selection is symmetric. Let $s$ be the side of $B$ crossed by edges of $\Gamma(G)$. In what follows, we view the intersection points along $s$, if any, as (dummy) feature points. A \emph{column} of $\Gamma(G)$ is a maximal set of feature points with the same $x$-coordinate. If $n$ and $k$ are the number of vertices and bends in $\Gamma(G)$, then there exist at most $n+k$ distinct columns. Order the columns within $B$ from the one containing the feature points closest to $s$ towards the one whose feature points are farthest away from $s$. Let $C$ and $C'$ be any two consecutive columns such that their horizontal distance is $w > \frac{\epsilon}{n+k}$. Observe that, between $C$ and $C'$, there is a vertical line $\ell$ that intersects no feature points but only a (possibly empty) set of horizontal edge-segments. Hence, we let $\sigma = \frac{\epsilon}{n+k}$ and apply a $(\sigma,\ell)$-strip removal operation. By repeating this procedure exhaustively we obtain the desired output $\Gamma'(G)$.
\end{proof}

\section{Initial Branching}
\label{sec:red}

In this section, we make the first step towards the fixed-parameter tractability of \TBOE\ by applying an initial branching step to simplify the problem---notably, this will allow us to focus on only extending the drawing inside a single face of $H$, and to assume that $H$ is an induced subgraph of $G$. 

We begin by introducing some additional notation that will be useful throughout the paper. 
Let $\langle (G, H, \Gamma(H)), \beta\rangle$  be an instance of \TBOE. A vertex $w\in V(H)$ is called an \emph{anchor} if it is incident to an edge in $E_X$. 
For a missing edge $vw\in E_X$ incident to a vertex $v\in V(H)$, we will use ``ports'' to specify a direction that $vw$ could potentially use to reach $v$ in an extension of $\Gamma(H)$; we denote these directions as $d$ which is an element from $\{\downarrow$ (north), $\uparrow$ (south), $\leftarrow$ (east), $\rightarrow$ (west)$\}$. Formally, a \emph{port candidate} for $vw\in E_X$ and $v\in V(H)$ is a pair $(v,d)$. A \emph{port-function} is an ordered set of port candidates which contains precisely one port candidate for each $vw\in E_X, v\in V(H)$, ordered lexicographically by $v$ and then by $w$.

We can now formalize the target problem that we will obtain from \TBOE\ via our exhaustive branching, which will be the focus of our considerations in Sections~\ref{sec:prep}-\ref{sec:tw}.
 
 \medskip\noindent\fbox{%
  \parbox{0.95\textwidth}{
  \textsc{Bend-Minimal Orthogonal Extension on a Face (\FTBOE)}\\
  \textbf{Input:} A planar graph $G_f$; an induced subgraph $H_f$ of $G_f$ with $k=|X_f|$, where $X_f=V(G_f)\setminus V(H_f)$; a drawing $\Gamma(H_f)$ of $H_f$ consisting of a single inner face $f$; a port-function $\mathcal{P}$.\\
  \textbf{Task:} Compute the minimum $\beta$ for which a $\beta$-extension of $\Gamma(H_f)$ exists and such that (1) missing edges and vertices are only drawn in the face $f$ and (2) each edge $xa\in E_X$ where $a\in V(H)$ connects to $a$ via its port candidate defined by $\mathcal{P}$, or determine that no such extension exists.
  }}\medskip

For the Turing reduction formalized in the next lemma, it will be useful to recall the definition of \TBOE\ and $\kappa$ from Section~\ref{se:preliminaries}.

\begin{lemma}\label{lem:reduction}
There is an algorithm that solves an instance $\mathcal I$ of \TBOE\ in time $3^{\bigoh(\kappa)} \cdot T(|I|,k)$, where $T(a,b)$ is the time required to solve an instance of \FTBOE\ with instance size $a$ and parameter value $b$.
\end{lemma}
\begin{proof}
Consider an instance $\langle (G, H, \Gamma(H)), \beta\rangle$ of \TBOE, and recall that $G$ is a planar graph of maximum vertex degree four. In addition, the graph $G$ may be assumed to be connected; if it is not, each connected component may be solved entirely independently. We begin with a simple observation. Let $\Gamma(G)$ be a planar drawing of $G$ and let $v$ be a vertex of $G$. Let $\Gamma_v(G)$ be the drawing obtained from $\Gamma(G)$ by removing $v$ and its incident edges. If $v$ has either one or two neighbors in $G$, then there are at most three faces of $\Gamma_v(G)$ that contain all neighbors of $v$ on each of their boundaries (because each neighbor of $v$ has degree at most three in $\Gamma_v(G)$). If $v$ has three or four neighbors, then by planarity there exists at most two faces containing all neighbors of $v$ on its boundary (in fact only one unless $\Gamma_v(G)$ is a cycle).

Based on this observation, for each missing vertex $x\in X$ adjacent to some vertex $y\in V(H)$, we branch to determine one of the at most three faces of $\Gamma(H)$ whose boundaries contain the vertices in $V(H)$ that are adjacent to $x$. By planarity, these are the only faces of $\Gamma(H)$ in which $x$ can be inserted in order to obtain a $\beta$-extension. In each branch, we obtain a specific face $f$ where $x$ is assumed to lie, and this in turn implies that all vertices in $X$ that admit a path to $x$ consisting solely of missing vertices and edges must also be placed in $f$. Since $G$ is connected, this implies that after at most $\kappa$ branching steps we will have assigned each missing vertex to some face of $\Gamma(H)$. 

Turning now to edges, each missing edge $ab\in E_X$ with at least one endpoint not in $V(H)$, say $a\in X$, must lie in the same face as $x$. To deal with edges $ab\not \in E(H)$ where $a,b\in V(H)$, we apply an additional branching step to determine whether they are drawn with no bends or with at least one bend. In the former case, we simply check if it is possible to extend the drawing $\Gamma(H)$ by a straight-line drawing of that edge (and if not, we discard the branch). In the latter case, we alter $G$ by subdividing the edge once and adding the newly created vertex $u_{ab}$ to $V$ while reducing the target value of $\beta$ by $1$. It is easy to observe that under the assumption of $ab$ being drawn with at least one bend, the resulting instance of \TBOE\ is equivalent to the original one. At that point, we naturally repeat the branching step for vertices in $X$ described in the previous paragraph to determine a face for $u_{ab}$. Crucially, this results in $H$ being an induced connected subgraph of $G$ while reducing the target value of $\beta$ by some fixed offset.

Finally, for each missing edge $vx$ where $v\in V(H)$ and $x\in V_X$ we exhaustively branch over all port candidates $p(v,d)$ for $d\in \directions$, restricting the choices to only those ports that are incident to the face we previously guessed for $x$. Note that there are at most three such ports for each missing edge, since $H$ is connected and therefore at least one port per vertex is already used by an edge of $H$. This procedure yields $3^{\bigoh(\kappa)}$ instances of \FTBOE, each of which contains $k\leq \kappa$ many missing vertices and $n_f\leq n$ many vertices in total. In particular, the output of \TBOE is true if and only if at least one branch leads to a set of \FTBOE\ instances such that if we sum up the integers returned by solving these instances plus the offset defined in the previous paragraph, we obtain a number that is at most $\beta$.
\end{proof}
  
\noindent We note that the marked face $f$ can be either the single inner face of $\Gamma(H_f)$ or the outer face. On a different note,  while \TBOE\ was stated as a decision problem for complexity-theoretic purposes, the output for \FTBOE is either an integer or ``No''. Two instances of \FTBOE are said to be \emph{equivalent} if their outputs are the same.
Note that checking whether an instance of \FTBOE\ admits some $\beta$-extension can be done in polynomial time by using the algorithm in~\cite{AngeliniRS21}. The pre-drawn graph given as input to the algorithm in~\cite{AngeliniRS21} will be $\Gamma(H_f)$ with a slight modification: if a vertex $v$ makes an angle larger than $\frac{\pi}{2}$  in the non-marked face $g$ of $\Gamma(H_f)$, then we add dummy vertices and connect them to $v$ until all angles around $v$ in $g$ are $\frac{\pi}{2}$. This guarantees that a solution only draws missing vertices inside the marked face $f$ (and not in $g$). Hence, we will assume to be dealing with instances where such an extension exists, and the task is to identify the minimum value of $\beta$. We will call a $\beta$-extension minimizing the value of $\beta$ a \emph{solution}.

\section{Preprocessing}
\label{sec:prep}

We can now focus on solving an instance of \FTBOE\ with only a single marked face $f$ being of interest. The aim of this section is to make the first two steps that will allow us to solve \FTBOE. This includes pruning out certain parts of the face which are provably irrelevant, and reducing the case of $f$ being the outer face to the case of $f$ being an inner face.

\subsection{Pruning}

Let $\Gamma(G)$ be an orthogonal drawing of a graph $G$ and let $f$ be a face of $\Gamma(G)$. A \emph{reflex corner} $p$ of $f$ is a feature point that makes an angle larger than $\pi$ inside $f$. Also, if $p$ is an anchor, then it is called an \emph{essential} reflex corner. A projection $\ell$ of a reflex corner $p$ is a horizontal or vertical line-segment in the interior of $f$ that starts at $p$ and ends at its first intersection with the boundary of $f$. \cref{fig:clean} (left) shows two projections $\ell_1$ and $\ell_2$ of a reflex corner $p$. 

Observe that each projection $\ell$ of a reflex corner $p$ divides the face $f$ into two connected regions, which are themselves orthogonal polygons. If $p$ is not essential and one of the two regions contains no reflex corners of its own (notice that inside this region, $p$ needs no longer be a reflex corner) and no anchors, we call the region \emph{redundant}. Our aim will be to show that such regions can be safely removed from the instance. More formally, recall that $\ell$ intersects the boundary of $f$ in $p$ on one side and in an element $e$ that is either a vertex $u$ or a point $q$ on an edge of $H$ on the other side of $f$. The \emph{pruning operation} at $\ell$ for a redundant region $\iota$ works as follows.  (1) If both $p$ and $e$ are vertices (which are therefore vertically or horizontally aligned) we add the edge $pu$ into $H$, whose representation in $\Gamma(H)$ is $\ell$. (2) If $p$ is a vertex and $e$ is an edge, we modify $H$ by replacing $q$ with a dummy vertex $v_q$ that subdivides $e$ and by adding the edge $pv_q$ (whose representation in $\Gamma(H)$ is $\ell$). (3) If $p$ is part of an edge $e'$ and $e$ is also an edge, we modify $H$ by replacing $p$ and $q$ with two dummy vertices $v_p$ and $v_q$ that subdivide $e'$ and $e$ and by adding the edge $v_pv_q$ (whose representation in $\Gamma(H)$ is $\ell$). We finally remove the  boundary of $\iota$ from $H$ and $\Gamma(H)$, except for the edge-segment $\ell$ and its end-vertices.

\begin{figure}[t]
    \centering
    \begin{minipage}[t]{0.31\textwidth}
    \centering
        \includegraphics[page=1,width=\textwidth]{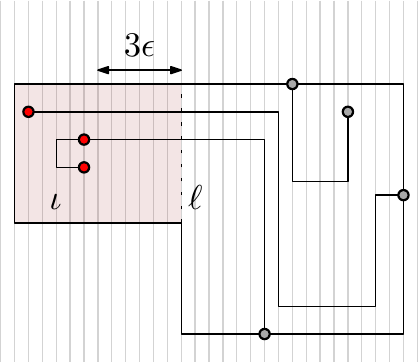}
    \end{minipage}\hfill
    \begin{minipage}[t]{0.31\textwidth}
    \centering
        \includegraphics[page=2,width=\textwidth]{figs/pruning.pdf}
    \end{minipage}\hfill
    \begin{minipage}[t]{0.31\textwidth}
    \centering
        \includegraphics[page=3,width=\textwidth]{figs/pruning.pdf}
    \end{minipage}
    \caption{Illustration for the proof of \cref{lem:singlecut}.}
    \label{fig:pruning}
\end{figure}
\begin{lemma}
\label{lem:singlecut}
Let $\mathcal{I}=\langle G_f, H_f, \Gamma(H_f), f, \mathcal{P} \rangle$ be an instance of \textsc{F-TBOE}. Let $\ell$ be a projection of some non-essential reflex corner in $f$, which gives rise to a redundant region $\iota$. Then pruning $\iota$ at $\ell$ results in an instance $\mathcal{I}^\iota$ that is equivalent to $\mathcal{I}$.
\end{lemma}

\begin{proof}
One direction is easy, since a solution of $\mathcal{I}^\iota$  can be readily transformed into a solution of $\mathcal{I}$ by undoing the pruning operation at $\ell$. Hence, consider a solution of $\mathcal{I}$. Let $\Gamma(G_f)$ be the orthogonal drawing of $G_f$ corresponding to the solution and let $\Gamma^\iota$ be the orthogonal drawing formed by the missing vertices and edge-segments of $\Gamma(G_f)$ in $\iota$.  

If $\Gamma^\iota$ is empty, then we can directly apply the pruning operation at $\ell$ and obtain a solution for $\mathcal{I}^\iota$. Thus, suppose that $\Gamma^\iota$ is not empty. 

We assume that $\ell$ is vertical, the case in which it is horizontal is analogous. Refer to \cref{fig:pruning} for an illustration. Let $\epsilon>0$ be a value such that any feature point has horizontal distance larger than $3\epsilon$ from $\ell$. Also, let $B$ be a rectangle such that: (a) one of its vertical sides $s$ is contained in $\ell$, (b) its opposite vertical side $s'$ is in $\iota$ and $s'$ is not intersected by any edge-segment, (c) it contains the whole drawing $\Gamma^\iota$. By construction, $B$ is a v-selection, and by \cref{le:box} we can compute a new \sequiv solution in which the $B$-selected drawing, namely $\Gamma^\iota$, has width $\epsilon$. 
Finally, since the reflex corner $p$ emanating $\ell$ is not essential, it cannot be connected to any vertex in $\Gamma_\iota$, thus we can translate $\Gamma_\iota$ by $2\epsilon$ such that it moves to the other side of $\ell$, i.e., it lies in the interior of the non-redundant region defined by $\ell$. The obtained orthogonal drawing has no feature point in $\iota$, and hence we are again in the position to apply the pruning operation at $\ell$ so to obtain a solution for $\mathcal{I}^\iota$.
\end{proof}

We can show that exhaustively applying \cref{lem:singlecut} results in an instance with the following property: each projection of each non-essential reflex corner in $f$ splits $f$ into two faces, each of which has at least one port on its boundary. We call such instances \emph{clean}; see \cref{fig:clean}.


\begin{figure}[t]
    \centering
    \begin{minipage}[t]{0.3\textwidth}
    \centering
        \includegraphics[page=4,width=\textwidth]{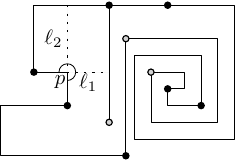}
    \end{minipage}\hfil
    \begin{minipage}[t]{0.3\textwidth}
    \centering
        \includegraphics[page=5,width=\textwidth]{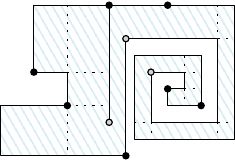}
    \end{minipage}\hfil
    \begin{minipage}[t]{0.3\textwidth}
    \centering
        \includegraphics[page=6,width=\textwidth]{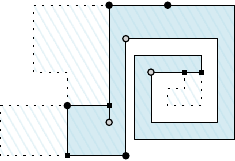}
    \end{minipage}
    \caption{Left: A reflex corner $p$ and its projections $\ell_1$ and $\ell_2$. Middle: A face (striped) with all its non-essential reflex corners and projections (anchor vertices have a gray filling while non-anchors are solid). Right: The corresponding clean instance (dummy vertices are drawn as small squares).}
    \label{fig:clean}
\end{figure}

\begin{lemma}
\label{lem:globalcut}
There is a polynomial-time algorithm that takes as input an arbitrary instance of \textsc{F-TBOE} and outputs an equivalent instance which is clean.
\end{lemma}

\begin{proof}
Let $\mathcal{I}=\langle G_f, H_f, \Gamma(H_f), f, \mathcal{P} \rangle$ be an instance of \textsc{F-TBOE}. 
It suffices to observe that every time we apply a pruning operation the number of projections strictly decreases. Namely, if this claim holds, we can iteratively look for a projection in the current instance and apply a pruning operation at such projection, until our instance becomes clean. Since we have at most three projections for each reflex corner, the algorithm performs a number of pruning operations that is linear in the size of $\Gamma(G_f)$.

To prove the claim, let $\ell$ be a projection at a reflex corner $p$ in an orthogonal drawing $\Gamma(H_f)$. Let $\Gamma'(H_f)$ be the orthogonal drawing obtained by applying the pruning operation at $\ell$. Suppose first that $p$ makes an angle of $\frac{3\pi}{2}$ in the marked face $f$. Then in $\Gamma'(H_f)$, point $p$ in not a reflex corner anymore, as it makes angle that is either $\frac{\pi}{2}$ or $\pi$. Moreover, the pruning operation does not introduce any new reflex corner in $\Gamma'(H_f)$, hence $\Gamma'(H_f)$ has one projection less than $\Gamma(H_f)$. Consider now the remaining case in which $p$ forms an angle $2\pi$ in $f$. Then there exist three projections at $p$ in $f$, and $\Gamma'(H_f)$ contains at most two of them. Again the pruning operation does not introduce any new reflex corner in $\Gamma'(H_f)$, hence $\Gamma'(H_f)$ has at least one projection less than~$\Gamma(H_f)$.
\end{proof}

Given \cref{lem:globalcut}, we will hereinafter assume that our instances of \textsc{F-TBOE} are clean.

\subsection{Outer Face}

Given an instance of \FTBOE\ where the marked face $f$ is the outer face of $\Gamma(H_f)$, let us begin by constructing a rectangle that bounds $\Gamma(H_f)$ and will serve as a ``frame''~for~any~solution.

\begin{observation}\label{ob:frame}
Let $\mathcal{I}=\langle G_f, H_f, \Gamma(H_f), f, \mathcal{P} \rangle$ be an instance of \FTBOE\ and let $R$ be a rectangle that contains $\Gamma(H_f)$ in its interior. Then $\mathcal{I}$ admits a solution that lies in the interior of $R$.
\end{observation}

\begin{proof}
\begin{figure}[t]
    \centering
    \begin{minipage}[t]{0.42\textwidth}
    \centering
        \includegraphics[page=2,width=\textwidth]{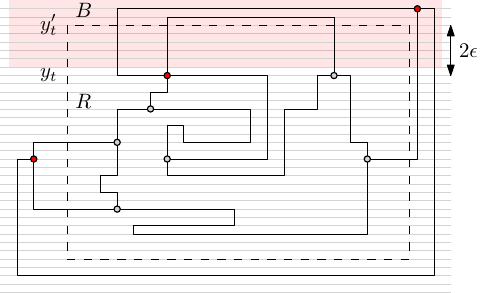}
        \subcaption{$\Gamma(G_f)$\label{fi:of-1}}
    \end{minipage}\hfill
    \begin{minipage}[t]{0.42\textwidth}
    \centering
        \includegraphics[page=3,width=\textwidth]{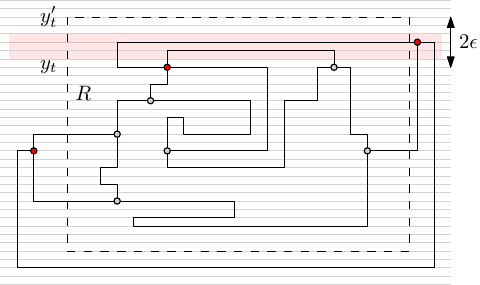}
        \subcaption{$\Gamma'(G_f)$\label{fi:of-2}}
    \end{minipage}\vspace{0.3cm}
    \begin{minipage}[t]{0.42\textwidth}
    \centering
        \includegraphics[page=4,width=\textwidth]{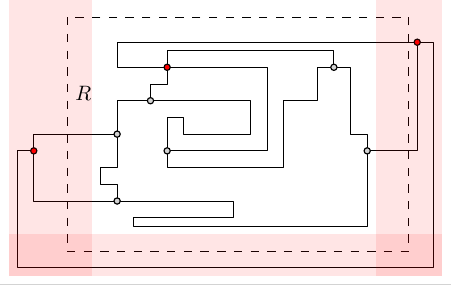}
        \subcaption{\label{fi:of-3}}
    \end{minipage}\hfill
    \begin{minipage}[t]{0.42\textwidth}
    \centering
        \includegraphics[page=5,width=\textwidth]{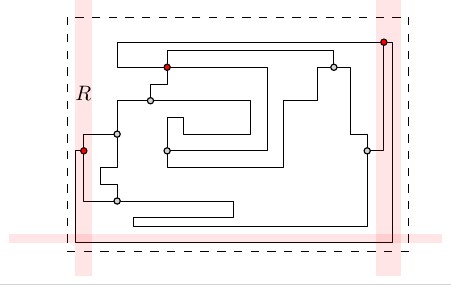}
        \subcaption{\label{fi:of-4}}
    \end{minipage}
    \caption{Illustration for the proof of \cref{ob:frame}.}
    \label{fig:outerface}
\end{figure}
Consider an orthogonal drawing $\Gamma(G_f)$ representing a solution to $\mathcal{I}$ that is not contained in the rectangle $R$. We first deal with the part of $\Gamma(G_f)$ that overflows above $R$; refer to \cref{fi:of-1} for an illustration. Let $y_{t}$ be the topmost coordinate of $\Gamma(H_f)$ and let $y'_{t}$ be the topmost coordinate of $R$. Let $\epsilon = \frac{y'_{t}-y_{t}}{2}$. Consider a rectangle $B$ such that: (a) its bottommost horizontal side is slightly above $y_{t}$ and does not contain any feature point, (b) no feature point of $\Gamma(G_f)$ is above/to the left/to the right of $B$. Then $B$ is an h-selection and by \cref{le:box} we can scale-down its $B$-selected drawing to have height at most $\epsilon$; see \cref{fi:of-2}. With a similar argument we can deal with the parts of $\Gamma(G_f)$ that overflow to the left of, to the right of, and below $R$; see \cref{fi:of-3,fi:of-4}. The original and final drawings are \sequiv and hence have the same number of bends.
\end{proof}

Based on \cref{ob:frame}, we shall assume that any instance $\mathcal{I}$ is modified such that the outer face of $\Gamma(H_f)$ is a rectangle $R$ containing no anchors (e.g., with four dummy vertices at its corners connected in a cycle). Notice that, while this ensures that $f$ is no longer the outer face, $f$ now contains a hole (that is, $H_f$ is not connected anymore). The aim for the rest of this section is to remove this hole by connecting it to the boundary of $R$. 

To do so, let us consider an arbitrary horizontal or vertical line-segment $\zeta$ that connects the boundary of $R$ with an edge-segment in the drawing $\Gamma(H_f)$ and intersects no other edge-segment of $\Gamma(H_f)$. Observe that, w.l.o.g., we can assume that each edge-segment in a solution $\Gamma(G_f)$ only intersects $\zeta$ in single points (and not in a line-segment); otherwise, one may shift $\zeta$ by a sufficiently small $\epsilon$ to avoid such intersections.
Roughly speaking, our aim will be to show that the instance $\mathcal{I}$ can be ``cut open'' along $\zeta$ to construct an equivalent instance where the boundary of the polygon includes $R$, and to branch in order to determine how the edges in a hypothetical solution cross through $\zeta$. However, to do so we need to ensure that there is a solution, in which the number of such crossings through $\zeta$ is bounded. 

Let us consider the drawing of a missing edge $e\in E_X$ in $\Gamma(G_f)$. The intersection points of $e$ with $\zeta$ partition the drawing of $e$ into polylines $e^{\zeta}_1$, $e^{\zeta}_2$, \dots, $e^{\zeta}_q$, where each pair of consecutive polylines $e^{\zeta}_i$ and $e^{\zeta}_{i+1}$ touch $\zeta$ at a point, which we denote by $z_i$ ($i=1,\dots,q-1$). We distinguish two cases depending on the structure of these polylines. A polyline $e^{\zeta}_j$, $1<j<q$, is called a $\zeta$-\emph{handle} if the unique region of the plane enclosed by $e^{\zeta}_j$ and $\zeta$ does not contain $\Gamma(H_f)$; otherwise the polyline is called a $\zeta$-\emph{spiral}. See \cref{fig:outerface2} for an illustration.




\begin{figure}[t]
    \centering
    \begin{minipage}[t]{0.42\textwidth}
    \centering
        \includegraphics[page=1,width=\textwidth]{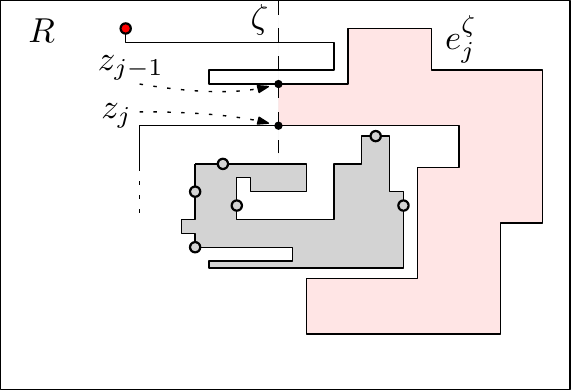}
    \end{minipage}\hfill
    \begin{minipage}[t]{0.42\textwidth}
    \centering
        \includegraphics[page=2,width=\textwidth]{figs/outerface-2.pdf}
    \end{minipage}
    \caption{Illustration of $\zeta$-handles (left) and $\zeta$-spirals (right).}
    \label{fig:outerface2}
\end{figure}

\begin{lemma}\label{le:handles}
Assume $\mathcal{I}$ and $\zeta$ are fixed as above. Then $\mathcal{I}$ admits a solution such that no missing edge contains a $\zeta$-handle.
\end{lemma}

\begin{proof}
\begin{figure}[t]
    \centering
    \begin{minipage}[t]{0.32\textwidth}
    \centering
        \includegraphics[page=3,width=\textwidth]{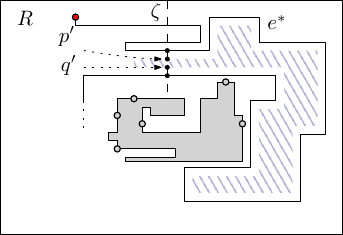}
    \end{minipage}\hfil
    \begin{minipage}[t]{0.32\textwidth}
    \centering
        \includegraphics[page=4,width=\textwidth]{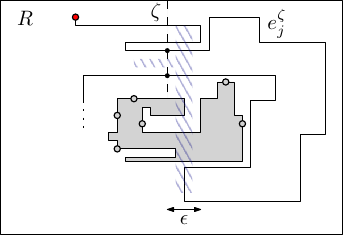}
    \end{minipage}\hfil
    \begin{minipage}[t]{0.32\textwidth}
    \centering
        \includegraphics[page=5,width=\textwidth]{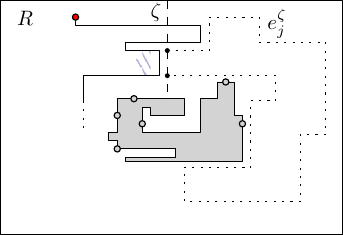}
    \end{minipage}
    \caption{Illustration for the proof of \cref{le:handles}.}
    \label{fig:zhandles}
\end{figure}
Let $\Gamma(G_f)$ be a solution to $\mathcal{I}$ and let $e^*$ be a $\zeta$-handle. Observe that, by planarity, the polyline $e^*$ is not crossed by any edge (except possibly at common endpoints). Consider the subdrawing $\Gamma_\zeta$ of $\Gamma(G_f)$ formed by all vertices and edge-segments in  the interior of the unique region of the plane enclosed by $e^*$ and $\zeta$. 

If $\Gamma_\zeta$  is empty, we can safely redraw $e^*$ with two bends (hence without increasing its total number of bends) such that it does not intersect $\zeta$ anymore. Thus, suppose $\Gamma_\zeta$ is not empty.

Assume $\zeta$ is vertical, as the argument is analogous if $\zeta$ is horizontal; refer to    \cref{fig:zhandles} for an illustration. Let $\epsilon>0$ be such that no feature point of $\Gamma(G_f)$ is at horizontal distance smaller than $\epsilon$ from $\zeta$. Let $p$ and $q$ be the top and bottom intersection points of $e^*$ with $\zeta$, respectively. Also, let $p'$ and $q'$ between $p$ and $q$ (if any) be the intersection points made by edges that are part of $\Gamma_\zeta$  that are closest to $p$ and $q$, respectively. First, we define a rectangle $B$ such that one vertical side corresponds to $\zeta$ and contains $\Gamma_\zeta$ in its interior. Rectangle $B$ is a v-selection for $\Gamma_\zeta$ whose $B$-selected drawing coincides with the whole $\Gamma_\zeta$. (On the other hand, $B$ may not be a valid v-selection for $\Gamma(G_f)$.) We apply \cref{le:box} to scale-down $\Gamma_\zeta$ such that its width becomes $\frac{\epsilon}{2}$. Note that the current drawing now may not be planar anymore. Consider a new rectangle $B'$ such that: (a) its bottommost horizontal side is between $p$ and $p'$ (or between $p$ and $q$ if $p'$ does not exist) and contains no feature point; (b) its topmost horizontal side is above any feature point of $\Gamma_\zeta$; (c) its leftmost and rightmost vertical sides are to the left and to the right, respectively, of any feature point of $\Gamma_\zeta$. By construction, $B'$ is an $h$-selection for $\Gamma_\zeta$ and we can apply \cref{le:box}  to scale-down the $B$-selected drawing such that its height is smaller than the vertical distance between $p$ and $p'$. Similarly, we can define and scale-down a rectangle $B''$ such that: (a) its topmost horizontal side is between $q'$ and $q$ (or between $p$ and $q$ if $q'$ does not exist) and contains no feature point; (b) its bottommost horizontal side is below any feature point of $\Gamma_\zeta$; (c) its leftmost and rightmost vertical sides are to the left and to the right, respectively, of any feature point of $\Gamma_\zeta$. The resulting drawing of $G_f$ is now planar. Indeed, $\Gamma_\zeta$ is now small enough that we can translate it horizontally on the other side of $\zeta$ without creating any edge crossing with other edges. After this operation, we are again in the position to redraw $e^*$ with two bends such that it does not intersect $\zeta$ anymore. 

By iterating this procedure for each $\zeta$-handle, we obtain the desired solution.
\end{proof}

Next we deal with $\zeta$-spirals: while they cannot be completely avoided, we show that one can bound the number of $\zeta$-spirals for each edge by a function of the parameter $k$.

\newcommand{\fspirals}{4k(k+1)}
\newcommand{\fspiralsperedge}{k+1}
\begin{lemma}\label{le:spirals}
Assume $\mathcal{I}$ and $\zeta$ are fixed as above. Then $\mathcal{I}$ admits a solution with no $\zeta$-handles and at most $\fspirals$ $\zeta$-spirals.
\end{lemma}
\begin{proof}
\begin{figure}[t]
    \centering
    \begin{minipage}[t]{0.42\textwidth}
    \centering
        \includegraphics[page=6,width=\textwidth]{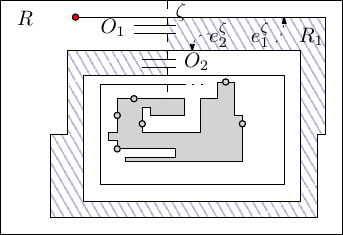}
    \end{minipage}\hfil
    \begin{minipage}[t]{0.42\textwidth}
    \centering
        \includegraphics[page=7,width=\textwidth]{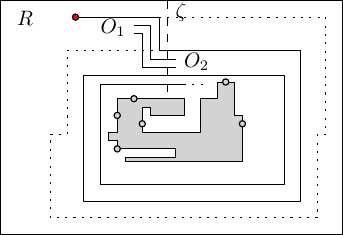}
    \end{minipage}\hfil
    \caption{Illustration for the proof of \cref{le:spirals}.}
    \label{fig:zspirals}
\end{figure}
By \cref{le:handles}, we know that $\mathcal{I}$ admits a solution $\Gamma(G_f)$ with no $\zeta$-handles. Observe that, by its definition, $\zeta$ can be crossed only be the edges in $E_X$ that are drawn inside $f$, which are at most $4k$. Let $\sigma$ be the greatest number of $\zeta$-spirals made by an edge crossing $\zeta$. If $\sigma \le \frac{\fspirals}{4k}=\fspiralsperedge$, we are done. Hence, suppose $\sigma > \fspiralsperedge$. 

Let $e$ be an edge containing a number of $\zeta$-spirals larger than  $\fspiralsperedge$. Let $e^\zeta_1,e^\zeta_2,\dots,e^\zeta_q$ be the $\zeta$-spirals made by $e$ such that $e^\zeta_i$ an $e^\zeta_{i+1}$ are consecutive ($i=1,2,\dots,q-1$), i.e., they touch at a point on $\zeta$. Consider two consecutive $\zeta$-spirals, $e^\zeta_i$ and $e^\zeta_{i+1}$ and the region $R_i$ of the plane bounded by $\zeta$ and by $e^\zeta_i$ and $e^\zeta_{i+1}$; refer to \cref{fig:zspirals} for an illustration. If $R_i$ contains a missing vertex in its interior, we call it \emph{filled}. Otherwise $R_i$ is not filled, and its interior is either completely empty or it contains some edge-segments. 

Suppose first that $R_i$ is completely empty. Then we can simply redraw edge $e$ by replacing $e^\zeta_i$ with a vertical segment between its endpoints, which we can slightly move to the left (or the right) so that it is not contained in $\zeta$. 

Suppose now that $R_i$  contains some edge-segments (but no vertices because it is not filled). In this case, any edge-segment in $R_i$ is part of a $\zeta$-spiral made by some other edge $e'$. (Otherwise the drawing would either be not connected or contain a $\zeta$-handle.) Then let $O_1$ be the ordered sequence of $\zeta$-spirals that we encounter when walking along $\zeta$ between the first and the last endpoint of $e^\zeta_i$, and similarly let $O_2$ be the ordered sequence of $\zeta$-spirals that we encounter when walking along $\zeta$ between the first and the last endpoint of $e^\zeta_{i+1}$. Let $O^*_1$ and $O^*_2$ be the sequences obtained from $O_1$ and $O_2$ by replacing each $\zeta$-spiral with the corresponding edge it belongs to (note that no edge appears more than once in any of the two sequences). We claim that $O^*_1$ and $O^*_2$ are identical. If this is true,  then we can apply, for each edge in $O^*_1=O^*_2$, a  rerouting operation analogous as the one described above (see \cref{fig:zspirals}). As a result, again $e$ can be redrawn so that it contains one less $\zeta$-spiral. To see that $O^*_1$ and $O^*_2$ are identical, suppose for a contradiction they are not. Then either they contain the same edges in different order, or they contain different edges. In the first case, two edges would cross each other contradicting planarity. In the second case, one edge would have one end-vertex in $R_i$ hence contradicting the fact that $R_i$ is not filled. 

By iterating this procedure we obtain a drawing in which each edge $e$ either contained at most $\fspiralsperedge$ $\zeta$-spirals per edge since the beginning and hence has not been modified, or any region defined by two consecutive $\zeta$-spirals of $e$ is filled. Since there exist at most $k$ such regions per edge, again $e$ contains at most $\fspiralsperedge$ $\zeta$-spirals, as desired.
\end{proof}

With \cref{le:spirals}, we obtain that there exists a solution where the total number of edge-segments crossing through $\zeta$ is at most $\fspirals$. We can use this to branch on which edges cross through $\zeta$ and use this to make a ``bridge'' connecting $R$ to the hole in $f$, thus resulting in an equivalent instance where $f$ is modified to become an inner face with no holes.

\begin{lemma}
There is an algorithm that takes as input an instance $\mathcal{I}$ of \FTBOE where~$f$ is the outer face and solves it in time $2^{\bigoh(k^2 \log k)}\cdot Q(|\mathcal{I}|,k)$, where $Q(a,b)$ is the time to solve an instance of \FTBOE\ with instance size $a$ and parameter value $b$ such that~$f$~is~the~inner~face.
\end{lemma}
\begin{proof}
Let $\mathcal{I}=\langle G_f,H_f,\Gamma(H_f),f \rangle$ be an instance of \FTBOE such that $f$ is an inner face bounded by a rectangle $R$ and containing a segment $\zeta$ defined as above.

By \cref{le:spirals}, it is not restrictive to consider solutions such that  each missing edge drawn in $f$ contains no $\zeta$-handles and at most $\fspirals$ $\zeta$-spirals. That is, we shall consider solutions in which $\zeta$ is crossed at most $\fspirals$ times. 

Denote by $E_{X_f} = (E(G_f) \cap E(H_f))$, that is, the subset if missing edges that are drawn in $f$. Assume $E_{X_f}$ be arbitrarily ordered. Also, recall that this set contains at most $4k$ edges. We encode the $\fspirals$ potential crossings on $\zeta$ that appear in a solution to $\mathcal I$ as an array $C_\zeta$. Namely, for each $i \in [\fspirals]$, $C_\zeta[i]=j$ means that the $i$-th crossing along $\zeta$ from top to bottom (assuming $\zeta$ is a vertical segment up to a temporary rotation of the drawing) is used by the $j$-th edge in $E_{X_f}$, while  $\C_\zeta[i]=\diamond$ means that no edge actually uses that potential crossing. Observe that there are $2^{\bigoh(k^2 \log k)}$ such arrays.

Fix an array $C_\zeta$ and modify the boundary of $f$ (and hence $\Gamma(H_f)$) as follows. Add two dummy vertices at the endpoints of $\zeta$, thus subdividing the two edges these two endpoints lie on. Add the edge-segment $\zeta$ and subdivide it $\fspirals$ times such that any two consecutive subdivision vertices are equispaced. Let $z_1,z_2,\dots,z_{\fspirals}$ be the subdivision vertices along $\zeta$ ordered from top to bottom (always assuming $\zeta$ is vertical). Consider any edge $uv$ of $E_{X_f}$ whose index $j$ appears at least once in $C_\zeta$. Let $i_1,i_2, \dots, i_q$ be the indexes such that $C_\zeta[i_1]=C_\zeta[i_2]=\dots=C_\zeta[i_q]=j$. If $q>1$, then $e$ forms $q-1$ $\zeta$-spirals, which we encode in the instance by modifying $G_f$ as follows. 

We further guess whether when walking along edge $uv$ from $u$ to $v$, in a hypothetical solution, the first crossing encountered is $z_{i_1}$ or not.  In the former case, we replace $uv$ with the path $u-z_{i_1}-z_{i_q}-z_{i_2}-z_{i_q-1}-\dots-v$. In the latter case, we replace $uv$ with the path $v-z_{i_1}-z_{i_q}-z_{i_2}-z_{i_q-1}-\dots-u$.  This gives $2^{\bigoh(k^2)}$ branches for each array. Also, observe that we have generated $2^{\bigoh(k^2 \log k)}$ new instances of \FTBOE in which $f$ is an inner face without holes. We now argue that $\cal I$ admits a solution if and only if at least one of these transformed instances does. 

One direction is easy, namely consider a transformed instance $\cal I'$ which admits an orthogonal drawing $\Gamma'$ as a solution. Undoing the transformation yields a solution to $\cal I$. More precisely, replacing all dummy vertices along $\zeta$ with inner points of the corresponding edges and removing the edge-segments along $\zeta$ yields an orthogonal drawing of $G_f$ with the same number of bends as $\Gamma'$. 

Suppose now that $\cal I$ admits a solution $\Gamma(G_f)$. Consider the transformed instance $\cal I'$ whose array $C_\zeta$ suitably encodes the crossings along $\zeta$ and in which each edge crossing $\zeta$ has been transformed in the correct path. In order to transform $\Gamma(G_f)$ into a solution $\Gamma'$ of $\cal I'$ transform each crossing of an edge $e$ with $\zeta$ into  a dummy vertex that subdivides $e$. Let $z'_1,z'_2,\dots,z'_q$ be the obtained dummy vertices and add the edges of the path $z'_1-z'_2-\dots-z'_q$ (which represent $\zeta$). It only remains to vertically align the dummy vertices such that $z'_i$ can be identified with $z_i$ ($i=1,2,\dots,q)$. 

We proceed as follows. We first define a rectangle $B$ whose topmost horizontal side coincides with the one of $R$ and whose bottommost horizontal side is slightly above the bottommost endpoint of $\zeta$ such that it contains no feature point. Since $B$ is an h-selection by construction, we apply \cref{le:box} and scale down the $B$-selected drawing such that its topmost feature point  is below the $y$-coordinate of $z_q$ and its height is smaller than the vertical distance between any two consecutive dummy vertices $z_i$ and $z_{i+1}$ ($i=1,2,\dots,q-1$). Now observe that $z'_1$ is also below the $y$-coordinate of $z_q$. Next, for $i=q,q-1,\dots,1$ (i.e., from bottom to top), we iterate the following procedure. By construction, $z'_i$ is below the $y$-coordinate of $z_i$, then we identify a horizontal line $\ell_i$ slightly below $z'_i$ such that it contains no feature point, and we apply a $(\sigma_i,\ell_i)$-strip addition (see \cref{prop:sigma}) such that the $y$-coordinate of $z'_i$ will coincide with the one of $z_i$.  This yields the desired solution $\Gamma'$ of $\mathcal{I}'$.
\end{proof}

\section{Discretizing the Instances}
\label{sec:discret}
Our next aim is to define the sector graph and show that it suffices to consider only a bounded number of possible points in each sector for extending $\Gamma(H_f)$. Essentially, this allows us to combinatorially extract those properties of $\Gamma(H_f)$ that are relevant for solving \FTBOE.

\subsection{Sectors and the Sector Graph}
 Recall that a \emph{port candidate} is a tuple $(a,d)$ where $a$ is an anchor and $d\in \directions$.

For a point $p\in f$, the \emph{bend distance} $\bdist(p, (a,d))$ to a port candidate $(a,d)$ is the minimum integer $q$ such that there exists an orthogonal polyline with $q$ bends connecting $p$ and $a$ in the interior of $f$ which arrives to $a$ from direction $d$.

\begin{definition}
Let $\mathcal{P}=((a_1,d_1),\dots,(a_q,d_q))$ be an ordered set of port candidates. For each point $p\in f$, we  define its bend-vector as the tuple $\bvect(p)=(\bdist(p, (a_1,d_1)),\dots,\bdist(p, (a_q,d_q)))$.
\end{definition}

\begin{definition}
Given an ordered set of port candidates $\mathcal{P}$, a \emph{sector} $F$ is a maximal connected set of points with the same bend-vector w.r.t.\ $\mathcal{P}$.
\end{definition}

When $\mathcal{P}$ is not specified explicitly, we will assume it to be the set of port candidates provided by the considered instance of \FTBOE.
The face $f$ is now partitioned into a set $\mathcal{F}$ of sectors. It is worth noting that sectors are connected regions in the face $f$ by the definition, which, in particular, can be degenerate: a sector may be a single point, or a line-segment.

At this point, we can define a graph representation capturing the adjacencies between the sectors in our instance; see \cref{fig:sgraph} for an illustration.

\begin{definition}
Sectors $A$ and $B$ are \emph{adjacent} if there exists a point $p$ in $A$ and a direction $d\in \directions$ such that the first point outside of $A$ hit by the ray starting from $p$ in direction $d$ is in $B$. 
\end{definition}

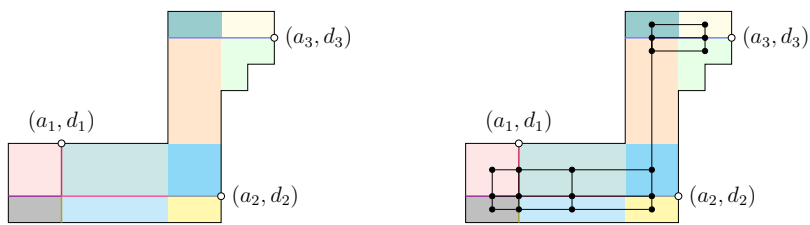
\begin{figure}[H]
\centering
\scalebox{0.7}{
\begin{tikzpicture}

\fill[fill=red!10!white] (0,.5) rectangle (1,2);
\fill[fill=gray!50!white] (0,0.5) rectangle (1,1);
\fill[fill=cyan!20!white] (1,.5) rectangle (3,1);
\fill[fill=cyan!40!white] (3,1) rectangle (4,2);
\fill[fill=yellow!40!white] (3, .5) rectangle (4,1);
\fill[fill=orange!20!white] (3, 2) rectangle (4,4);
\fill[fill=green!10!white] (4, 3) rectangle (5,4);
\fill[fill=yellow!10!white] (4, 4) rectangle (5,4.5);
\fill[fill=teal!20!white] (1, 1) rectangle (3,2);
\fill[fill=teal!40!white] (3, 4) rectangle (4,4.5);
\fill[fill=white!10!white] (4.5, 3) rectangle (5,3.5);
\draw
(0, .5) -- (1.5, .5) -- (2.5, .5) -- (4, .5) -- (4, 3)  -- (4.5, 3) -- (4.5, 3.5) -- (5, 3.5) -- (5, 4.5) -- (3, 4.5) -- (3, 2) -- (0, 2) -- (0, .5);

\draw[thick] (1.01, 1) -- (0.99, 1);
\draw[thick, blue!50!white] (5, 4) -- (3, 4);
\draw[thick, purple!80!white] (1, 2) -- (1, 1);
\draw[thick, violet!80!white] (0, 1) -- (1, 1);
\draw[thick, olive!80!white] (1, 1) -- (1, .5);
\draw[thick, magenta!80!white] (1, 1) -- (3, 1);
\draw[thick, gray!80!white] (4, 1) -- (3, 1);

\node[circ, label=above:{\Large $(a_1, d_1)$}] (v) at (1, 2) {};
\node[circ, label=right:{\Large $(a_2, d_2)$}] (v) at (4, 1) {};
\node[circ, label=right:{\Large $(a_3, d_3)$}] (v) at (5, 4) {};

\end{tikzpicture}
}
\hspace{1 cm}
\scalebox{0.7}{
\begin{tikzpicture}

\fill[fill=red!10!white] (0,.5) rectangle (1,2);
\fill[fill=gray!50!white] (0,0.5) rectangle (1,1);
\fill[fill=cyan!20!white] (1,.5) rectangle (3,1);
\fill[fill=cyan!40!white] (3,1) rectangle (4,2);
\fill[fill=yellow!40!white] (3, .5) rectangle (4,1);
\fill[fill=orange!20!white] (3, 2) rectangle (4,4);
\fill[fill=green!10!white] (4, 3) rectangle (5,4);
\fill[fill=yellow!10!white] (4, 4) rectangle (5,4.5);
\fill[fill=teal!20!white] (1, 1) rectangle (3,2);
\fill[fill=teal!40!white] (3, 4) rectangle (4,4.5);
\fill[fill=white!10!white] (4.5, 3) rectangle (5,3.5);
\draw
(0, .5) -- (1.5, .5) -- (2.5, .5) -- (4, .5) -- (4, 3)  -- (4.5, 3) -- (4.5, 3.5) -- (5, 3.5) -- (5, 4.5) -- (3, 4.5) -- (3, 2) -- (0, 2) -- (0, .5);

\draw[thick] (1.01, 1) -- (0.99, 1);
\draw[thick, blue!50!white] (5, 4) -- (3, 4);
\draw[thick, purple!80!white] (1, 2) -- (1, 1);
\draw[thick, violet!80!white] (0, 1) -- (1, 1);
\draw[thick, olive!80!white] (1, 1) -- (1, .5);
\draw[thick, magenta!80!white] (1, 1) -- (3, 1);
\draw[thick, gray!80!white] (4, 1) -- (3, 1);

\node[circ, label=above:{\Large $(a_1, d_1)$}] (v) at (1, 2) {};
\node[circ, label=right:{\Large $(a_2, d_2)$}] (v) at (4, 1) {};
\node[circ, label=right:{\Large $(a_3, d_3)$}] (v) at (5, 4) {};

\node[bcirc] (v1) at (.5, 1.5) {};
\node[bcirc] (v2) at (.5, 1) {};
\node[bcirc] (v3) at (.5, .75) {};
\node[bcirc] (v4) at (1, 1.5) {};
\node[bcirc] (v5) at (1, 1) {};
\node[bcirc] (v6) at (1, .75) {};
\node[bcirc] (v7) at (3.5, 1.5) {};
\node[bcirc] (v19) at (2, 1.5) {};
\node[bcirc] (v8) at (2, 1) {};
\node[bcirc] (v9) at (2, .75) {};
\node[bcirc] (v11) at (3.5, 4) {};
\node[bcirc] (v12) at (4.5, 4) {};
\node[bcirc] (v13) at (3.5, 4.25) {};
\node[bcirc] (v14) at (4.5, 4.25) {};
\node[bcirc] (v15) at (3.5, 3.75) {};
\node[bcirc] (v16) at (4.5, 3.75) {};
\node[bcirc] (v17) at (3.5, 1) {};
\node[bcirc] (v18) at (3.5, .75) {};
\draw (v1) -- (v2)
(v1) -- (v4)
(v2) -- (v5)
(v2) -- (v3)
(v3) -- (v6)
(v4) -- (v5)
(v5) -- (v6)
(v4) -- (v7)
(v5) -- (v8) 
(v6) -- (v9)
(v8) -- (v9)
(v8) -- (v19)
(v7) -- (v11)
(v11) -- (v12)
(v9) -- (v18)
(v18) -- (v17)
(v17) -- (v8)
(v17) -- (v7)
(v15) -- (v16)
(v13) -- (v14)
(v14) -- (v16)
(v13) -- (v15)
;
\end{tikzpicture}
}
\caption{Left: partioning a face $f$ into a set $\cal F$ of sectors, with three anchors marked using white circles. Right: the  graph representation of $\cal F$.
    \label{fig:sgraph}}
\end{figure}

Observe that the relationship of being adjacent is symmetric; furthermore, for a specific direction $d$ we say that sector $A$ is \emph{$d$-adjacent} to $B$ if $A$ is adjacent to $B$ for this choice of $d$. The \emph{sector graph} $\mathcal{G}$ is the graph whose vertex set is the set of sectors $\mathcal{F}$, and adjacencies of vertices are defined via the adjacency of sectors.

It will be useful to establish some basic properties of the sector graph. For instance, it is not difficult to observe that the sector graph is a connected planar graph. Furthermore, we can show that the boundary between two sectors is, in a sense, simple. Concerning its size, we observe that each sector contains at least one intersection point between two projections and that any such intersection point can be shared by at most nine sectors (four non-degenerate sectors plus five degenerate sectors). Hence:

\begin{observation}
\label{obs:sectorbound}
The number of vertices in $\calG$ is upper-bounded by $9x^2$, where $x$ is the number of feature points in $\Gamma(H_F)$.
\end{observation}





\subsection{The Sector-Grid}

A property of sectors that will become important later is that, inside each sector, we only need  a bounded number of positions for the placement of feature points in a hypothetical solution. In particular, our aim will be to construct a ``universal'' point-set with the property that there exists a solution which places feature points only on these points, and where the intersection of the point-set with each sector is upper-bounded by a function of the parameter.
Before we construct such a universal point set, we will first need to subdivide sectors into ``subsectors'' which have grid-like connections to each other. Crucially, we will show that the number of subsectors in each sector is upper-bounded by a function of $k$.

Let us fix a sector $S$ and a direction $d\in \directions$, say w.l.o.g.\ $d=\rightarrow$. Let a reflex corner be \emph{critical} if it is incident to at least two distinct sectors, and $(S,d)$-\emph{critical} if it is critical and also can be reached by a ray from some point in $S$ traveling in direction $d$. To construct the subsectors of $S$, let us project all $(S,d)$-critical reflex corners (for all four choices of $d$) into $S$ to obtain a grid, and make each induced grid cell in $S$ a \emph{subsector} of $S$. 
Observe that for each subsector in each sector $S$, it holds that its entire boundary in each direction is either the boundary of $f$, or touches the boundary of a single other ``adjacent'' subsector (which may or may not belong to $S$).

Crucially, we show that the number of such subsectors obtained from each sector is not too large. This will be important when using sectors for dynamic programming in Section~\ref{sec:tw}, since it will allow us to bound the size of the universal point set in each sector.

\begin{lemma}
\label{lem:subsectors}
For each $S$, $d$, there are at most $4k$ $(S,d)$-\emph{critical} reflex corners.
\end{lemma}
\begin{proof}
As before, let us describe the case where $d=\rightarrow$ whereas the remaining three cases are completely symmetric. By definition, a projection of each $(S,d)$-critical reflex corner $c$ in at least one of the three directions other than $d$ (i.e., $\uparrow$, $\downarrow$ or $\leftarrow$) must run along the boundary of some two sectors, say $A$ and $B$, incident to $c$. This in particular means that there is a port (say $(x,d')$) which distinguishes $A$ from $B$; without loss of generality, assume that the bend distance from $A$ to $(x,d')$ is $\psi$ while the bend distance from $B$ to the same port is $\psi+1$. 

Consider the case where the projection of $c$ separating $A$ from $B$ is in the direction $\leftarrow$. In that case, we observe that the bend distance to $(x,d')$ differs for every pair of vertically aligned points $c^\uparrow$, $c^\downarrow$ which are placed directly above and below the projection of $c$, respectively. In particular, this also holds if we place $c^\uparrow$ and $c^\downarrow$ directly inside the sector $S$, hence contradicting the assumption that $S$ is a sector.

This means that every $(S,d)$-critical reflex corner must be \emph{vertical}, i.e., must be incident to the same sector in the $\nwarrow$ and $\swarrow$ directions, but be incident to a different sector in at least one of the two remaining diagonal directions. Moreover, consider a vertical $(S,d)$-critical reflex corner such that the sector $A$ (i.e., the sector with the smaller bend distance to the port $(x,d')$ distinguishing $A$ from $B$) has a larger $x$ coordinate than $c$. In that case, the projection of $c$ in the direction $\leftarrow$ once again behaves as in the horizontal case: the bend distance to $(x,d')$ differs for every pair of vertically aligned points $c^\uparrow$, $c^\downarrow$ which are placed directly above and below this projection of $c$, respectively. In particular, this once again also holds if we place $c^\uparrow$ and $c^\downarrow$ inside the sector $S$, contradicting the assumption that $S$ is a sector.

Our aim is now to show that there can be at most $k$ vertical $(S,d)$-critical reflex corners such that the sector $A$ has a smaller $x$ coordinate than $c$. To this end, let us consider up to $q$ such reflex corners $c_1,\dots,c_q$ ordered from the smallest $x$ coordinate to the largest. For each such vertical $(S,d)$-critical reflex corner $c_i$, we observe that the port $(x_i,d_i)$ distinguishing its sector $A_i$ from $B_i$ must be different from the port of distinguishing the sectors $A_j$ from $B_j$ in every reflex corner $c_j$, $j>i$---indeed, points that are on the same horizontal coordinate but directly to the left and right of the vertical projection of $c_j$ can reach the vertical projection of $c_i$ with the same number of bends, which here implies that they have the same bend distance to $(x_i,d_i)$. The proof now follows by $4k$ upper-bounding the total number of ports.
\end{proof}

By applying Lemma~\ref{lem:subsectors} on all sides of each sector $S$, we obtain that $S$ is partitioned into at most $(8k)^2$ subsectors. Observe that we may refine the sector graph constructed earlier by partitioning sectors into subsectors, with adjacencies between subsectors defined in the same way as between sectors. Note that by definition, each pair of adjacent subsectors share the complete side of the boundary that connects them. Hence, we can define a \emph{subsector-column} as a set of subsectors which form a path in the subsector graph and span the same vertical strip in $\Gamma(H_f)$, and similarly a \emph{subsector-row} is a set of subsectors which forms a path in the subsector graph and span the same horizontal strip in $\Gamma(H_f)$.

With the above in mind, we proceed to build the universal point set. As our first step, we construct an auxiliary set of points we call a \emph{skeleton}.
Let us now choose an arbitrary horizontal line-segment for each subsector-row that intersects it, and similarly an arbitrary vertical line-segment for each subsector-column that intersects it. To construct the skeleton, for each subsector $v$, we define the point $p_v$ to be the point at the intersection of the two line-segments intersecting the subsector.

\newcommand{\gridsize}{\texttt{subgridsize}}
\newcommand{\newgridsize}{\texttt{gridsize}}
\newcommand{\fgrid}{112k^3+202k^2+85k}
Let $\gridsize(k)=\fgrid$. We place a set of $\gridsize(k)\times \gridsize(k)$ points in a grid-like arrangement into each subsector $v$, where the points are centered at $p_v$ and the grid underlying these points occupies a square area of $\epsilon \times \epsilon$ for a sufficiently small $\epsilon$. In particular, we choose $\epsilon$ to be sufficiently small so that a horizontal or vertical projection of any pair of grid points intersects with the same line-segment of $\Gamma(H_f)$. We call this point set $\mathcal{S}_v$ the \emph{subsector-grid} of a subsector $v$; in the degenerate cases where $v$ is a line-segment or single point, the subsector-grid is a set of points on that segment or just a single point, respectively.

 Towards proving that using one subsector-grid for each subsector gives a sufficiently large point-set to embed the missing subgraph, we begin with a technical lemma.

\begin{longlemma} \label{le:tech}
Let $\zeta$ be a horizontal or vertical line-segment such that it lies in the interior of $f$, except for its endpoints that are on the boundary of $\Gamma(H_f)$, and such that it does not intersect any other edge-segment of $\Gamma(H_f)$. There is a solution in which any missing edge $e$ crosses $\zeta$  at most $k$ times. 
\end{longlemma}
\begin{proof}
The proof adopts similar arguments as in \cref{le:handles,le:spirals}, although in a slightly different setting. Let $\Gamma(G_f)$ be a solution and let $e$ be an edge that crosses $\zeta$ more than $k$ times. By symmetry, we can assume $\zeta$ is vertical. Consider a walk along the polyline representing $e$ from one end-vertex to the other, and let $(s_1,s_2,\dots,s_h)$ be the ordered sequence of edge-segments traversed along this walk. Observe that only the horizontal edge-segments of this sequence can cross $\zeta$, hence it must be $h>2k$.  Also, let $z_1,z_2,\dots,z_q$ be the crossings between $e$ and $\zeta$ ordered from top to bottom. 

Next, consider two horizontal segments $s_i$ and $s_j$, with $1<i<j<h$ and $j-i>2$. Let $e_{ij}$ be the polyline $(s_i,s_{i+1},\dots,s_j)$. Denote by $z_{i^*}$ ($z_{j^*}$) the crossing point between $z_i$ ($z_j$) and $\zeta$. Also, denote by $p_i$ and $p_j$ the endpoints of $e_{ij}$. We say that $s_i$ and $s_j$ can be \emph{shortcut} if $j^*=i^*+1$, and  $p_i$ and $p_j$ are both to the left or both to the right of $\zeta$.  Then an argument analogous to \cref{le:handles} shows that an equivalent solution exists in which $e_{ij}$  can be redrawn such that all edge segments $s_l$ with $i<l<j$ are vertical  and $e^*_{ij}$ does not cross $\zeta$ anymore. 

Based on this, we shall assume that $e$ contains no pair of segments that can be shortcut. The crossing points along $\zeta$ partitions $e$ into $q+1$ polylines, called $\zeta$-\emph{pieces}, which we denote by $e^\zeta_1,e^\zeta_2,\dots,e^\zeta_{q+1}$, such that $e^\zeta_i$ an $e^\zeta_{i+1}$ touch at a point on $\zeta$. Observe that when walking along $e$, the crossing points are visited in the order $z_1,z_q,z_2,z_{q-1},\dots,z_{\lceil q/2 \rceil}$ (or its reverse). By observing that $\zeta$-pieces behave analogously as $\zeta$-spirals, we can follow the lines of the proof of \cref{le:spirals}. Consider two consecutive  $\zeta$-pieces $e^\zeta_i$ and $e^\zeta_{i+1}$, and the region $R_i$ of the plane bounded by $\zeta$ and by $e^\zeta_i$ and $e^\zeta_{i+1}$. If $R_i$ contains a missing vertex in its interior, we call it \emph{filled}. Otherwise $R_i$ is not filled, and its interior is either completely empty or it contains some edge-segments. Suppose first that $R_i$ is completely empty. Then we can simply redraw $e^*$ such that in the new instance the number of crossings along $\zeta$ is one less and no new bends are introduced (but some are removed). Suppose now that $R_i$  contains some edge-segments (but no vertices because it is not filled). In this case, any edge-segment in $R_i$ is part of a $\zeta$-piece made by some other polyline $e'$. (Otherwise the drawing would either be not connected or a pair that can be shortcut.) Then let $O_1$ be the ordered sequence of $\zeta$-pieces that we encounter when walking along $\zeta$ between the first and the last endpoint of $e^\zeta_i$, and similarly let $O_2$ be the ordered sequence of $\zeta$-pieces that we encounter when walking along $\zeta$ between the first and the last endpoint of $e^\zeta_{i+1}$. Let $O^*_1$ and $O^*_2$ be the sequences obtained from $O_1$ and $O_2$ by replacing each $\zeta$-piece with the corresponding edge it belongs to (note that no edge appears more than once in any of the two sequences). We claim that $O^*_1$ and $O^*_2$ are identical. If this is true,  then we can redraw each edge in $O^*_1=O^*_2$, and again redraw $e^*$  so that it crosses $\zeta$ one less time.  By iterating this procedure we obtain a drawing in which the number of crossings of $e$ along $\zeta$ is at most $k$ and the overall number of bends did not increase.

We conclude by observing that, if $e$ crosses $\zeta$ more than once, then one of its endpoints is enclosed in  a region of the plane bounded by a $\zeta$-piece $e^\zeta$ and the line-segment between the two crossing points that $e^\zeta$ makes with $\zeta$. This observation will be useful in the proof of next lemma.
\end{proof}

\begin{lemma}
\label{lem:sectorgrid}
There exists a solution such that each feature point not in $\Gamma(H_f)$ lies on a subsector-grid point of some subsector.
\end{lemma}
\begin{proof}
Consider a solution $\Gamma(G_f)$. 
Let $\mathcal{A}$ be a subsector-column. Let $\Gamma_\mathcal{A}$ be the drawing formed by all edge-segments and vertices that do not belong to $\Gamma(H_f)$ and that lie in $\mathcal{A}$. We first argue that if a polyline representing (part of) an edge in $\Gamma_\mathcal{A}$ contains a large number of bends, then we can redraw it and obtain an equivalent solution $\Gamma'(G_f)$. This holds vacuously if the subsectors in $\mathcal{A}$ are line-segments or  single points, hence suppose this is not the case.

\begin{figure}[t]
    \centering
    \begin{minipage}[t]{0.3\textwidth}
    \centering
        \includegraphics[page=3,width=\textwidth]{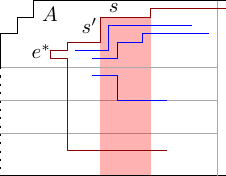}
        
    \end{minipage}\hfil
    \begin{minipage}[t]{0.3\textwidth}
    \centering
        \includegraphics[page=4,width=\textwidth]{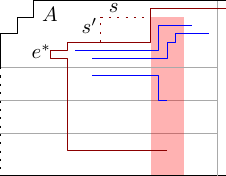}
    \end{minipage}\vspace{0.3cm}
    \caption{Illustration for the proof of \cref{lem:sectorgrid}.
    \label{fig:staircase}}
\end{figure}
Let $e^*$ be a polyline in $\Gamma_\mathcal{A}$ representing (part of) and edge $e$, and let $p_1$ and $p_2$ be its two endpoints. Observe that each $p_i$ ($i=1,2$) may be a vertex that lies in $\mathcal{A}$ or on its boundary, or an inner point of a longer polyline that enters $\mathcal{A}$ from another subsector-column. Consider a walk along $e^*$ from $p_1$ to $p_2$ and let $(s_1,s_2,\dots,s_h)$ be the sequence of edge-segments of $e^*$ ordered according to this walk. For each pair of consecutive edge-segments along this walk, we either make a left turn or a right turn (and hence a bend). We call  a \emph{staircase} a maximal sequence of edge-segments for which we alternate left and right turns.  

Let $e^*_s=(s_c,s_{c+1},\dots,s_d)$ be a staircase in $e^*$ with at least four segments, and let $s_i$ and $s_{i+1}$ be two consecutive edge-segments internal to $e^*_s$, that is, $i=c+1,d-1$. Without loss of generality, we can assume $s_i$ is horizontal and $s_{i+1}$ is vertical. Also, we say that $s_i$ and $s_{i+1}$ form an \emph{upward-step}, if the point shared between $s_i$ and $s_{i+1}$ is the topmost endpoint of $s_{i+1}$, while $s_i$ and $s_{i+1}$ form a \emph{downward-step} otherwise. A vertex $v$ \emph{blocks} pair $(s_i,s_{i+1})$ (and the pair $(s_i,s_{i+1})$ is  \emph{blocked}), if: (a) its $x$-coordinate falls within the horizontal range of $s_i$ (including its endpoints),  (b) it is below $s_{i+1}$ in case of an upward-step  or above $s_i$ in case of a downward-step, (c) it the end-vertex of a missing edge. Note that if $v$ blocks $(s_i,s_{i+1})$, then it is either an anchor or a missing vertex. On the other hand, for a pair of segments that is not blocked, we can apply a simple redrawing technique illustrated in \cref{fig:staircase}, which merges $s_i$ and $s_{i+1}$ with $s_{i-1}$ and $s_{i+2}$, respectively. As a consequence, we shall assume that each staircase contained in a polyline $e^*$ in $\Gamma_\mathcal{A}$ contains only blocked pairs. 

Next, let $\zeta$ be a horizontal or vertical line-segment having its endpoints on some edge-segments of $\Gamma(H_f)$ and that does not intersect any other edge-segment of $\Gamma(H_f)$. Also, choose $\zeta$ such that it maximizes the number of intersections with $e^*$.  If $\zeta$ crosses $e^*$ only once, then $e^*$ is a staircase. Also, no vertex can block more than two pairs of consecutive edge-segments, else there would be a line $\zeta'$ intersecting $e^*$ in two points, which would contradict the choice of $\zeta$. Then $e^*$ contains at most $10k$ blocked pairs, therefore it has at most $10k+2$ bends. Hence, suppose that $\zeta$ crosses $e^*$ more than once. By \cref{le:tech}, we shall assume that $e^*$ crosses $\zeta$ at most $k$ times. Also, at the end of the proof of \cref{le:tech}, we observed that if $e^*$ crosses $\zeta$ more than once, then one of its endpoints is enclosed in  a region of the plane bounded by a $\zeta$-piece $e^\zeta$ and the line-segment between the two crossing points that $e^\zeta$ makes with $\zeta$. Hence, there can be at most two line-segments, $\zeta$ and $\zeta'$, that cross $e^*$ more than once and on different edge-segments. 

Then $e^*$ can be partitioned in at most $k+1$ $\zeta$-pieces and $k+1$ $\zeta'$-pieces. Each of these pieces contains at least $2$ bends, which further partition it into $3$ smaller polylines that run between the endpoints of $e^*$ and these bends. Since we can always avoid pairs of edge-segments that can be shortcut, each of these polylines either contains at most $3$ bends or it is a staircase. Thus, in total, $e^*$ contains at most $18(k+1)+10k+2=28k+20$ bends. 

The next question towards proving our statement is what is the maximum number of disjoint polylines in $\Gamma_\mathcal{A}$ that represent the same edge $e$. By \cref{le:tech}, one verifies that there is always an equivalent solution in which this number is at most $\lceil k/2 \rceil$. Thus, in total, a single edge may contribute with at most $(k+1)(28k+20)$ bends in $\Gamma_\mathcal{A}$. (In fact, this is an over estimation because if an edge enters in the same subsector-column more than once, then each piece is a staircase in the worst case.)

Since a subsector-columnn contains at most $|X|=k$ vertices and $|E_X| \le 4k$ polylines, the total number of feature points in $\Gamma_\mathcal{A}$ is at most $k+4k(k+1)(28k+20)$. Also, for each polyline in $\Gamma_\mathcal{A}$, we further consider two feature points for its endpoints (which might not be vertices but inner points). This adds at most $4k(k+1)$ additional feature points, hence we have $k+4k(k+1)(28k+20)+4k(k+1)=\fgrid=\gridsize(k)$. A symmetric argument proves that in a subsector-row we need at most $\gridsize(k)$ points. 

Consider now any subsector-column $\mathcal{A}$. Recall that $\Gamma_\mathcal{A}$ partitions its vertices into columns, that is, maximal sets of vertices with the same $x$-coordinate (but different $y$-coordinates). By the argument above, we know that the number of different columns is at most $\gridsize(k)$ and we can order them from left to right. Next, we assign to each feature point $p$ in $\Gamma_\mathcal{A}$ a number $x_p$ equal to the rank of $p$'s column. By repeating this procedure for all subsector-columns (in any order), we have that all feature points $p$ receive a number $x_p$. 

Consider now any subsector-row $\mathcal{B}$ and its drawing $\Gamma_\mathcal{B}$. $\Gamma_\mathcal{B}$ partitions its vertices  into rows, that is, maximal sets of vertices with the same $y$-coordinate (but different $x$-coordinates). Again, we know that the number of different rows is at most $\gridsize(k)$ and we can order them from top to bottom. Next, we assign to each feature point $p$ in $\Gamma_\mathcal{B}$ a number $y_p$ equal to the rank of $p$'s row. By repeating this procedure for all subsector-rows (in any order), we have that all feature points $p$ receive a number $y_p$. 

Finally, we map each feature point $p$ of $\Gamma(G_f)$ to the point of the subsector grid of its subsector at the intersection between the $x_p$-th column of the grid and the $y_p$-th row of the grid. Since the  horizontal and vertical order between pairs of feature points is preserved, as well as horizontal and vertical colinearities, one verifies that the resulting drawing and $\Gamma(G_f)$ are \sequiv. 
\end{proof}

From Lemmas~\ref{lem:subsectors},~\ref{lem:sectorgrid}
and by setting $\newgridsize(k)=\gridsize(k)^2\cdot (8k)^2$, we obtain:

\begin{corollary}
\label{cor:grid}
Given an instance $\mathcal{I}$ of \FTBOE we can construct a point set (called a \emph{sector grid}) in time $\bigoh(|\mathcal{I}|)$ with the following properties: \textup{(1)} $\mathcal{I}$ admits a solution whose feature points all lie on the sector grid, and \textup{(2)} each sector contains at most $\newgridsize(k)$ points of the sector grid.
\end{corollary}

\section{Exploiting the Treewidth of Sector Graphs}
\label{sec:tw}
In this section, we complete the proof of our fixed-parameter tractability result by first showing that the sector graphs in fact have treewidth bounded by a function of the parameter $k$, and then by using this fact to design a dynamic programming algorithm solving \FTBOE.

\subsection{Sector Graphs Are Tree-Like}
We begin by introducing some notation that will be useful in this subsection. Let $\mathcal{P}=((a_1,d_1),\dots,(a_q,d_q))$ be the ordered set of port candidates for the considered face $f$. Also, $q\leq 4k$, because the degree of the vertices being added is at most $4$. For each $1\leq i\leq q$, let $\mathcal{P}_i=((a_1,d_1),\dots,(a_i,d_i))$ be a prefix of length $i$ of $\calP$. 
For each $1\leq i\leq q$, we denote by $\mathcal{F}_i$ and $\calG_i$ the set of sectors and the sector graph, respectively, obtained by considering the bend distances to $\mathcal{P}_i$.
Using this terminology, we obtain that the graph $\mathcal{G}_q$ is precisely the sector graph of our initial instance, which we will also simply denote as $\mathcal{G}$. Furthermore, for a sector $F\in V(\calG_t)$ we denote by $\calU_{F}^{t+1}$ the set of sectors in $\calG_{t+1}$ that $F$ is partitioned into when one additionally considers bend distances to $(a_{t+1},d_{t+1})$; in other words, $\calU_{F}^{t+1}$ is the unique set with the property that $\bigcup_{Q\in \calU_{F}^{t+1}} Q=F$.

\begin{lemma}
\label{lem:secttw-1}
The sector graph $\mathcal{G}_1$ is a tree.
\end{lemma}
\begin{proof}
We prove the claim by construction. First, observe that the $0$-bend sector is merely a ray extending out of the port $(a_1,d_1)$, and the $1$-bend sectors are obtained as projections of this ray in the two directions orthogonal to the ray. For $1\leq i \leq q$, let the $i$-\emph{interface} be the set of line-segments that touch one sector at bend distance $i$ on one side and points at bend distance $i-1$ on the other side. In the case of the $0$-bend sector, we set the $0$-interface to be the port. Crucially, for every $i>0$ we observe that each $i$-interface touches the boundary of $f$ at both of its endpoints---indeed, otherwise the given $i$-interface could be extended in the direction in which it does not touch the boundary of $f$.

Let the $j$-\emph{subface} be the union of all sectors with bend distance at most $j$.
For each integer $i$, all sectors at bend distance $i$ can be constructed from the $i$-interfaces by simply projecting them into the previously unprocessed part of $f$. Crucially, for any pair of distinct $i$-interfaces $p$, $q$, we note that when traversing the boundary of the $(i-1)$-subface between $p$ and $q$, we must intersect the boundary of the original face $f$ due to the observation concerning the endpoints of interfaces at the end of the previous paragraph. Hence, each sector at bend distance $i$ can only touch precisely one $i$-interface---indeed, it must touch at least one such interface by the connectivity of $f$, but if it were to touch two such interfaces there would exist a closed curve in $f$ from $a_1$ through one such interface, the given sector, the other interface and back to $a_1$ which would enclose a piece of the boundary of $f$.

Hence, we conclude that each sector at bend distance $i$ can only be adjacent to a single sector at bend distance $i-1$, and this rules out the existence of cycles in the sector graph.
\end{proof}

 Lemma~\ref{lem:secttw-1} will be used as a base of an inductive argument establishing a bound on the treewidth of $\calG$. See \cref{fig:sectors-example} for an example of the sectors for two port candidates. We start by considering how each sector $F\in \calF_t$ maps to a subset $\calU_{F}^{t+1}$ of sectors in $\calF_{t+1}$. Towards this aim, let us now consider an arbitrary sector $F\in \calF_t$ for some $1\leq t\leq q$. We say that a line segment $\delta$ on the boundary of $F$ is an $F$-\emph{baseline} if (1) each point in $F$ can be reached by a ray starting at and orthogonal to $\delta$, and (2) $\delta$ touches $F$ on one side and points in $f\setminus F$ on the other side. When $F$ is clear from context, we simply use baseline for brevity.

\begin{figure}
\centering
\scalebox{0.7}{
\begin{tikzpicture}

\fill[fill=red!10!white] (0,.5) rectangle (4,2);
\fill[fill=blue!10!white] (1.5, 0) rectangle (2.5,.5);
\fill[fill=blue!10!white] (3, 2) rectangle (4,4);
\fill[fill=green!10!white] (4, 3) rectangle (5,4);
\fill[fill=white!10!white] (4.5, 3) rectangle (5,3.5);
\draw[thin] 
(0, .5) -- (1.5, .5) -- (1.5, 0) -- (2.5,0) -- (2.5, .5) -- (4, .5) -- (4, 3)  -- (4.5, 3) -- (4.5, 3.5) -- (5, 3.5) -- (5, 4) -- (3, 4) -- (3, 2) -- (0, 2) -- (0, .5);

\draw [->] (0.9, 1.2) -- (.6, 1.2);
\draw [->] (1.1, 1.2) -- (1.4, 1.2);
\draw [->] (2, .4) -- (2, .1);
\draw [->] (3.5, 2.1) -- (3.5, 2.4);
\draw [->] (4.1, 3.5) -- (4.4, 3.5);

\draw[very thick, red!50!white]
(0.98, 2) -- (0.98, .5)
(1.02, 2) -- (1.02, .5);

\draw[very thick, blue!50!white]
(1.5, .5) -- (2.5, .5)
(3, 2) -- (4, 2);

\draw[very thick, green!50!white]
(4, 3) -- (4, 4);

\draw
(1, 2) -- (1, .5);

\node[circ, label=above:{\Large $(a_1, d_1)$}] (v) at (1, 2) {};

\end{tikzpicture}
}
\scalebox{0.7}{
\begin{tikzpicture}

\fill[fill=red!10!white] (0,.5) rectangle (4,2);
\fill[fill=red!10!white] (1.5, 0) rectangle (2.5,.5);
\fill[fill=red!10!white] (3, 2) rectangle (4,4);
\fill[fill=blue!10!white] (4, 3) rectangle (5,4);
\fill[fill=white!10!white] (4.5, 3) rectangle (5,3.5);

\draw[thin] 
(0, .5) -- (1.5, .5) -- (1.5, 0) -- (2.5,0) -- (2.5, .5) -- (4, .5) -- (4, 3)  -- (4.5, 3) -- (4.5, 3.5) -- (5, 3.5) -- (5, 4) -- (3, 4) -- (3, 2) -- (0, 2) -- (0, .5);

\draw [->] (2, 1.1) -- (2, 1.4);
\draw [->] (2, 0.9) -- (2, 0.6);
\draw [->] (4.1, 3.5) -- (4.4, 3.5);

\draw[very thick, red!50!white]
(0, 1.02) -- (4, 1.02)
(0, 0.98) -- (4, 0.98);

\draw[very thick, blue!50!white]
(4, 3) -- (4, 4);

\draw
(0, 1) -- (4, 1);

\node[circ, label=right:{\Large $(a_2, d_2)$}] (v) at (4, 1) {};

\end{tikzpicture}}
\scalebox{0.7}{
\begin{tikzpicture}

\fill[fill=red!10!white] (0,.5) rectangle (1,2);
\fill[fill=yellow!20!white] (0,0.5) rectangle (1,1);
\fill[fill=cyan!20!white] (1,.5) rectangle (4,1);
\fill[fill=blue!10!white] (1.5, 0) rectangle (2.5,.5);
\fill[fill=orange!20!white] (3, 2) rectangle (4,4);
\fill[fill=green!10!white] (4, 3) rectangle (5,4);
\fill[fill=teal!20!white] (1, 1) rectangle (4,2);
\fill[fill=white!10!white] (4.5, 3) rectangle (5,3.5);
\draw
(0, .5) -- (1.5, .5) -- (1.5, 0) -- (2.5,0) -- (2.5, .5) -- (4, .5) -- (4, 3)  -- (4.5, 3) -- (4.5, 3.5) -- (5, 3.5) -- (5, 4) -- (3, 4) -- (3, 2) -- (0, 2) -- (0, .5);

\draw[thick] (1.01, 1) -- (0.99, 1);
\draw[thick, purple!80!white] (1, 2) -- (1, 1);
\draw[thick, violet!80!white] (0, 1) -- (1, 1);
\draw[thick, olive!80!white] (1, 1) -- (1, .5);
\draw[thick, magenta!80!white] (1, 1) -- (4, 1);

\node[circ, label=above:{\Large $(a_1, d_1)$}] (v) at (1, 2) {};
\node[circ, label=right:{\Large $(a_2, d_2)$}] (v) at (4, 1) {};

\end{tikzpicture}
}
\caption{Sectors with respect to (a) the first port; (b) the second port; (c) $\calP_{2}$. For a sector of each color, the segment on the border highlighted with the same color is its baseline; for (c) different sectors have different colors, and notice that at the intersection of the rays from $(a_1, d_1)$ and $(a_2, d_2)$ there is also a single point sector.}
\label{fig:sectors-example}
\end{figure}
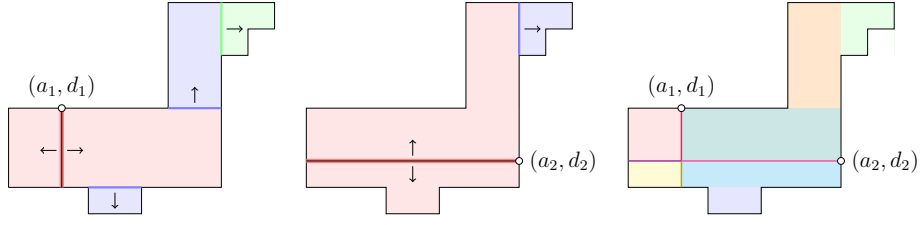

\begin{lemma}
\label{lem:sector-prorerty}
Each sector in $\calF_t$, $1\leq t \leq q$, admits at least one baseline.
\end{lemma}
\begin{proof}
We prove the claim by induction in the number of ports, where for $t=1$ the claim follows directly from the proof of Lemma~\ref{lem:secttw-1}, for which we observed that each sector with $\bdist=i$ is a projection of its $i$-interface (thus we can take it as the sector's baseline). For the inductive step, assume that for $\calP_{t-1}$, each sector $F\in \calF_{t-1}$ has an $F$-baseline. In particular, the unique sector $F\in \calF_{t-1}$ such that $F'\subseteq F$ has an $F$-baseline, say $\delta$, and hence a projection of $\delta$ also reaches every point on the boundary of $F'$. Recall that $F'$ is obtained as the intersection of $F$ with some particular sector with respect to bend distances only to a single port, which---as was argued in Lemma~\ref{lem:secttw-1}---is a projection of some interface $\omega$. In both cases of $\omega$ being either orthogonal or parallel to $\delta$, it is easy to see that the projection of $\delta$ reaches the boundary of $F'$ in only a single line segment, which establishes the claimed existence of an $F'$-baseline.
\end{proof}

The existence of a baseline is already quite helpful to obtain the desired bound on the treewidth, but not yet sufficient on its own. In particular, this implies that each sector has the shape of a histogram. Next, we show that the bend distances to any ``additional port'' cannot differ too much within a sector.

\begin{lemma}
\label{lem:recdepth}
For every sector $F\in \calF_{t}$, $t\in [1,q-1]$, and every pair $F_1, F_2\in \calU_{F}^{t+1}$, $|\bdist(p, (a_{t+1}, d_{t+1}))-\bdist(q, (a_{t+1}, d_{t+1}))|\leq 3$ for every pair of points $p\in F_1$, $q\in F_2$.
\end{lemma}
\begin{proof}
Assume for a contradiction that there are two sectors $F_1, F_2\in \calU_{F}^{t+1}$ and $p\in F_1$, $q\in F_2$, such that $|\bdist(p, (a_{t+1}, d_{t+1}))-\bdist(q, (a_{t+1}, d_{t+1}))|>3$.
Without loss of generality, let $\bdist(p, (a_{t+1}, d_{t+1}))<\bdist(q, (a_{t+1}, d_{t+1}))$. Then every polyline connecting $p$ to $q$ must have at least $3$ bends; if this is not the case, one could use this polyline to connect $q$ to $p$ with at most $2$ bends, and then use an additional bend at $p$ to obtain that $\bdist(p, (a_{t+1}, d_{t+1}))+3\geq \bdist(q, (a_{t+1}, d_{t+1}))$, which is a contradiction.

And yet, at the same time there always exists a polyline between $p$ and $q$ lying completely in $F$ that has at most $2$ bends. Indeed, by Lemma \ref{lem:sector-prorerty}, we can draw a polyline that starts at $p$ and proceeds to the baseline of $F$, uses a bend to reach another point on the baseline of $F$, and then performs one final bend to reach $q$. This results in a contradiction.
\end{proof}

\begin{figure}
\centering
\scalebox{0.5}{
\begin{tikzpicture}

\draw[thin] 
(0, 0) -- (11, 0)
(11, 2) -- (11, 0)
(11, 2) -- (10, 2)
(10, 2) -- (10, 1)
(9, 1) -- (10, 1)
(9, 1) -- (9, 5)
(8, 5) -- (9, 5)
(8, 5) -- (8, 3)
(7, 3) -- (8, 3)
(7, 3) -- (7, 2)
(6.5, 2) -- (7, 2)
(6.5, 2) -- (6.5, 4)
(5.5, 4) -- (6.5, 4)
(5.5, 4) -- (5.5, 2)
(5, 2) -- (5.5, 2)
(5, 2) -- (5, 3)
(4.5, 3) -- (5, 3)
(4.5, 3) -- (4.5, 2.5)
(4, 2.5) -- (4.5, 2.5)
(4, 2.5) -- (4, 4.5)
(2.5, 4.5) -- (4, 4.5)
(2.5, 4.5) -- (2.5, .5)
(2, .5) -- (2.5, .5)
(2, .5) -- (2, 2)
(2, 2) -- (1, 2)
(1, 2) -- (1, 1.5)
(1, 1.5) -- (.5, 1.5)
(.5, 1.5) -- (.5, .5)
 (.5, .5) --  (0, .5)
(0, .5) -- (0, 0)
;

\draw[very thick]
(0, 0) -- (11, 0);

\draw[very thick, red] 
(11, 2) -- (10, 2)
(8, 5) -- (9, 5)
(5.5, 4) -- (6.5, 4)
(4.5, 3) -- (5, 3)
(2.5, 4.5) -- (4, 4.5)
(2, 2) -- (1, 2)
;

\draw[very thick, blue] 
(2, .5) -- (2.5, .5)
(4, 2.5) -- (4.5, 2.5)
(5, 2) -- (5.5, 2)
(9, 1) -- (10, 1)
(6.5, 2) -- (7, 2)
;

\node (tau) at (.25, -.3) {\Large{$\delta$}};
\node (tau) at (-.25, .25) {\Large{$\alpha$}};

\end{tikzpicture}
}
\caption{The segments colored red (blue) are local maxima (minima). \label{fig:maxmin}}
\end{figure}
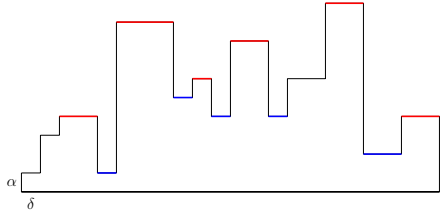

With Lemmas~\ref{lem:sector-prorerty} and~\ref{lem:recdepth}, we are ready to proceed to the most difficult part of establishing our bound on the treewidth of the sector graph. Let us fix some $F$-baseline $\delta$ for a sector $F$ in the sector graph $\calG_t$, $1\leq t\leq q$. Consider the polyline $\alpha$ obtained when traversing $F$ in clockwise fashion from one endpoint of $\delta$ to the other, where $\alpha$ does not intersect $\delta$. We call a line segment in $\alpha$ a \emph{local maximum} (\emph{minimum}) if $\alpha$ makes a right (left) turn both before and after the line segment (see Figure~\ref{fig:maxmin}). Let $\xi_{\max}(F)$ ($\xi_{\min}(F)$) denote the number of local maxima (local minima) in $F$; note that since each sector is a histogram, $\xi_{\max}(F)=\xi_{\min}(F)+1$.

\begin{figure}
\centering
\scalebox{0.5}{
\begin{tikzpicture}

\draw[thin] 
(0, 0) -- (11, 0)
(11, 2) -- (11, 0)
(11, 2) -- (10, 2)
(10, 2) -- (10, 1)
(9, 1) -- (10, 1)
(9, 1) -- (9, 5)
(8, 5) -- (9, 5)
(8, 5) -- (8, 3)
(7, 3) -- (8, 3)
(7, 3) -- (7, 2)
(6.5, 2) -- (7, 2)
(6.5, 2) -- (6.5, 4)
(5.5, 4) -- (6.5, 4)
(5.5, 4) -- (5.5, 2)
(5, 2) -- (5.5, 2)
(5, 2) -- (5, 3)
(4.5, 3) -- (5, 3)
(4.5, 3) -- (4.5, 2.5)
(4, 2.5) -- (4.5, 2.5)
(4, 2.5) -- (4, 4.5)
(2.5, 4.5) -- (4, 4.5)
(2.5, 4.5) -- (2.5, .5)
(2, .5) -- (2.5, .5)
(2, .5) -- (2, 2)
(2, 2) -- (1, 2)
(1, 2) -- (1, 1.5)
(1, 1.5) -- (.5, 1.5)
(.5, 1.5) -- (.5, .5)
 (.5, .5) --  (0, .5)
(0, .5) -- (0, 0)
;

\draw[dashed] 
(2.75, 4.5) -- (2.75, 0) 
(3.75, 4.5) -- (3.75, 0);

\draw[dotted]
(0, .5) -- (2, .5)
(4.5, 2.5) -- (5, 2.5)
(5.5, 2) -- (9, 2)
(10, 1) -- (11, 1)
;

\node (tau) at (3.25, 0.35) {\Large{$F_{\min}$}};

\end{tikzpicture}
}
\scalebox{0.5}{
\begin{tikzpicture}

\draw[thin] 
(0, 0) -- (11, 0)
(11, 2) -- (11, 0)
(11, 2) -- (10, 2)
(10, 2) -- (10, 1)
(9, 1) -- (10, 1)
(9, 1) -- (9, 5)
(8, 5) -- (9, 5)
(8, 5) -- (8, 3)
(7, 3) -- (8, 3)
(7, 3) -- (7, 2)
(6.5, 2) -- (7, 2)
(6.5, 2) -- (6.5, 4)
(5.5, 4) -- (6.5, 4)
(5.5, 4) -- (5.5, 2)
(5, 2) -- (5.5, 2)
(5, 2) -- (5, 3)
(4.5, 3) -- (5, 3)
(4.5, 3) -- (4.5, 2.5)
(4, 2.5) -- (4.5, 2.5)
(4, 2.5) -- (4, 4.5)
(2.5, 4.5) -- (4, 4.5)
(2.5, 4.5) -- (2.5, .5)
(2, .5) -- (2.5, .5)
(2, .5) -- (2, 2)
(2, 2) -- (1, 2)
(1, 2) -- (1, 1.5)
(1, 1.5) -- (.5, 1.5)
(.5, 1.5) -- (.5, .5)
 (.5, .5) --  (0, .5)
(0, .5) -- (0, 0)
;

\draw[dashed] 
(9, 4) -- (8, 4)
(2.5, 1.5) -- (9, 1.5)
(2.5, 2) -- (6.5, 2);

\draw[dotted]
(2.5, .5) -- (2.5, 0)
(9, 1) -- (9, 0)
(.5, .5) -- (2, .5)
(10, 1) -- (11, 1)
;

\node (tau) at (8.35, 2) {\Large{$F_{\min}$}};

\end{tikzpicture}
}
\caption{Cases of relative location of the $F_{\min}$ sector in $F$ relative to the $F$-baseline, Lemma~\ref{lem:subsect}. }\label{fig:hist}
\end{figure}
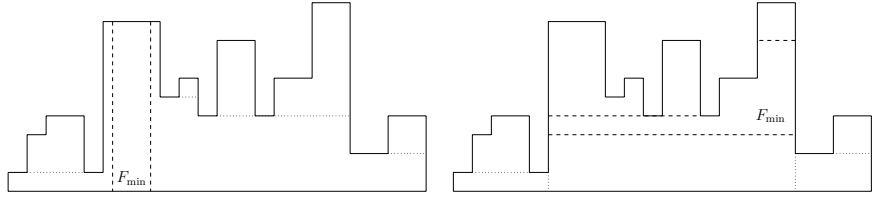

\begin{lemma}
\label{lem:subsect}
For every sector $F\in \calF_t$, $1\leq t \leq q-1$, we have $|\calU_{F}^{t+1}|\leq 4+\xi_{\max}(F)$ and $\max\limits_{F'\in\calU^{t+1}_F}\xi_{\max}(F')\leq \xi_{\max}(F)$.
\end{lemma}
\begin{proof}
Let $F_{\min}\in \calU_{F}^{t+1}$ be a sector with a minimum bend distance to $(a_{t+1},d_{t+1})$, say $m$. 
Since we have already observed that each sector has the shape of a histogram extending from some baseline (cf. Lemma \ref{lem:sector-prorerty}), let us now consider $2$ cases depending on whether the $F_{\min}$-baseline is parallel with the $F$-baseline or orthogonal to it (see Fig.~\ref{fig:hist}).

For the first case, observe that there are at most two line segments on the border of $F_{\min}$ that touch the baseline of $F$ precisely at a single point. Let us proceed under the assumption that there are precisely two such segments $\ell_1$, $\ell_2$ (the cases with less than $2$ segments follow analogously). Projecting $\ell_1$ and $\ell_2$ out of $F_{\min}$ into $F$ results in two sectors of $\calU_{F}^{t+1}$, say $F_{\ell_1}$ and $F_{\ell_2}$, with bend distance to $(a_{t+1},d_{t+1})$ that is one larger than that of $F_{\min}$. Since $F$ is a histogram, every point in $F\setminus (F_{\min}\cup F_{\ell_1} \cup F_{\ell_2})$ can be reached by projecting a line segment on the boundary of $F_{\ell_1}$ or $F_{\ell_2}$ into $F$, and in particular each such line segment gives rise to a separate sector in $\calU_{F}^{t+1}$. We conclude this case by observing that the number of line segments that may be projected from $F_{\ell_1}$ or $F_{\ell_2}$ into $F$ is upper-bounded by the number of local maxima.

For the second case, $F_{\min}$ splits $F$ into connected regions, where at most one such region intersects the $F$-baseline; again, we describe the more difficult case where the projection of a line segment on the boundary of $F_{\min}$ intersects the $F$-baseline. In this case, removing $F_{\min}$ partitions $F$ into connected regions, where at most one such region contains the $F$-baseline (we call this the \emph{base region}) and the local maxima of $F$ are partitioned between these connected regions. More precisely, each non-base region must contain at least one local maximum of $F$ and will be a single sector in $\calU_{F}^{t+1}$. The analysis of the base region then follows analogously as the previous case, where the projection of the side of $F_{\min}$ into this region behaves as a sector $F'_{\min}$ with minimum bend distance to $(a_{t+1},d_{t+1})$ within the base region.

Based on this construction, it is easy to see that also the number of local maxima of every sector in $\calU_{F}^{t+1}$ cannot exceed $\xi_{\max}(F)$.
\end{proof}

To obtain the main result of this section (Theorem~\ref{thm:secttw}), we will combine Lemma~\ref{lem:subsect} with the following lemma that bounds the number of local maxima in each sector.

\begin{lemma}
\label{lem:sectsum}
For each sector $F$ in $V(\mathcal{G})$, $\xi_{\max}(F) \leq 4k$.
\end{lemma}

\begin{proof}
We begin by establishing a bound on $\xi_{\max}(F')$ for each sector $F'$ arising from $\calP_1$. Observe that all local maxima in $F'$ must be parts of the border of the initial $f$ of the solving instance. Moreover, either $F'$ contains only a single local maximum, or traversing the boundary of $F'$ on at least one side of each local maximum in $F'$ leads to a reflex corner $\tau$; if this occurs on both sides of the local maximum, we let $\tau$ be the closer of the two reflex corners. Now consider the projection from $\tau$ that is parallel to the $F'$-baseline; this projection splits $F'$ into two connected regions, where the region containing the given local maximum is just a rectangle. Since there must be at least one port in each such rectangle, we immediately see that the number of local maxima in $F'$ is at most $4k$. 

The statement of the lemma now follows directly by applying \Cref{lem:subsect} $k$ times on each sector of $V(\mathcal{G}_1)$.
\end{proof}

\begin{theorem}
\label{thm:secttw}
Let $\mathcal{G}$ be a sector graph of a face $f$ of the drawing $\Gamma(G)$. Then $\tw(\mathcal{G})\leq (4+4k)^{4k}$.
\end{theorem}
\begin{proof}
We prove the claim by induction, where the base of an induction exactly follows from the result of the Lemma~\ref{lem:secttw-1}.
For the inductive step, assume that $\tw(\calG_t)=\calO(k^t)$ and our aim will be to show that $\tw(\calG_{t+1})$ is $\calO(k^{t+1})$.

Consider a tree decomposition $\calT_{\calG_t}=(\chi, T)$ of a graph $\calG_{t}$, and construct a tree decomposition $\calT_{\calG_{t+1}}=(\chi', T)$ for a graph $\calG_{t+1}$.
For each node $v\in T$, consider its bag $\chi(v)$. Then, to obtain $\chi'(v)$, each vertex $F\in\chi(v)$ will be replaced with the set of vertices $\calU_F^{t+1}$, i.e. $\chi'(v)=\bigcup\limits_{F\in \chi(v)}\calU_{F}^{t+1}$. All the properties of the tree decomposition are still preserved, since this operation merely replaces a vertex in $\chi(t)$ with a set of vertices in $\chi'(t)$ which always remain together; more precisely, $\calG_t$ can be obtained by contracting each $\calU_{F}^{t+1}$ into $F$.

Now, let us bound the treewidth of $\calT_{\calG_{t+1}}$. By Lemma~\ref{lem:subsect}, for each $F\in V(\calG_t)$, we have $|\calU_{F}^{t+1}|\leq 4+\xi_{\max}(F)$, and by Lemma~\ref{lem:sectsum} we have that $\xi_{\max}(F)\leq 4k$. Hence for each bag $t$, we have $|\chi'(t)|\leq (4+4k)\cdot |\chi(t)|$, which concludes the proof of the lemma.
\end{proof}

\subsection{The Final Step}

At this point, we have shown that an instance $\mathcal{I}=\langle G_f, H_f, \Gamma(H_f),\mathcal{P} \rangle$ with $k=|V(G_f)\setminus V(H_f)|$ of \FTBOE\ admits a sector graph $\mathcal{G}$ of treewidth at most $(4+4k)^{4k}$ (Theorem~\ref{thm:secttw}), and that a bend-minimal extension of $\Gamma(H_f)$ to an orthogonal planar drawing of $G_f$ can be assumed to only contain feature points on the sector-grid points as per Corollary~\ref{cor:grid}, of which there are at most $\newgridsize(k)$ many per sector. This allows us to proceed to the final ingredient for our algorithm:

\begin{lemma}
\label{lem:dynprog}
\FTBOE can be solved in time $2^{k^{\bigoh(1)}} \cdot |V(G_f)|$.
\end{lemma}
\begin{proof}
We begin by using the recent algorithm of Korhonen~\cite{Korhonen21} to compute a nice tree decomposition $(T,\chi)$ of $\calG$ of width at most $(4+k)^k$. Our aim will be to dynamically process $\calG$ along $T$ and enumerate all possible options of how a bend-minimal extension can intersect the sector-grid points of the sectors in the current bag. As our first step towards this aim, we formalize the records that will be stored in the dynamic program.

Recall that Lemma~\ref{lem:sectorgrid} guarantees the existence of a ``nice'' solution where all added feature points and added vertices lie on the sector-grid. There are at most $k^{\bigoh(\newgridsize(k))}$ possible ways any hypothetical nice solution may intersect with sector $v$, and these can be exhaustively enumerated by considering all possible placements of missing vertices on the sector-grid points of $v$ and all possible placements of missing edges on the line-segments projected out of the grid points. Let $\gridsols(v)$ be the set of all such possible intersections between a nice solution and $v$.

The core ingredient in our records is the notion of a \emph{configuration} of a node $t$ in $T$, which is a tuple $(X_V, X_E, \theta)$ where: 
\begin{itemize}
\item $X_V\subseteq V(G_f)\setminus V(H_f)$,
\item $X_E\subseteq E(G_f)\setminus E(H_f)$, and
\item $\theta$ maps each $v\in \chi(t)$ to an element of $\gridsols(v)$.
\end{itemize}
Intuitively, a configuration uses $X_V$ and $X_E$ to capture which of the missing vertices and edges have already been completely drawn in the sectors that have been processed so far, while $\theta$ captures how the solution intersects with the individual sectors in the bag.

To define our records at $t$, we will need to link the configurations introduced above to the ``partial solutions'' that can be constructed for the sectors in $\chi_\downarrow(t)$. To this end, let us consider an arbitrary configuration $\mathcal{C}=(X_V,X_E,\theta)$. Let the \emph{past bend number} of $\mathcal{C}$ be the minimum number of required to extend the orthogonal drawing of $\Gamma(H_f)$ in all sectors in $\chi_\downarrow(t)\setminus \chi(t)$ with the vertices in $X_v$, the edges in $X_E$, as well as all edges that $\theta$ routes to the boundaries of the sectors in $\chi_\downarrow(t)\setminus \chi(t)$, while respecting the port assignment $p$ in the \FTBOE\ instance. If no such extension exists, we simply set the \emph{past bend number} of $\mathcal{C}$ to ``$\infty$'', where $\infty+z=\infty$ for every $z$. 

With this, we are finally ready to formally define the dynamic programming records at node $t$: $\rec(t)$ is the mapping from all configurations at $t$ to their past bend numbers. 

Observe that since our nice tree decomposition has an empty root bag, there is only a single relevant configuration at the root $r$ of $T$---notably $(V(G_f)\setminus V(H_f), E(G_f)\setminus E(H_f), \emptyset)$. Hence, if we are given the records at the root $r$ we can solve \FTBOE\ by simply outputting $\rec(r)(V(G_f)\setminus V(H_f), E(G_f)\setminus E(H_f), \emptyset)$. Moreover, since all leaves contain an empty bag as well, the records for each leaf node $t$ are simply $\rec(t)=\{(\emptyset,\emptyset,\emptyset)\mapsto 0\}$. Hence, to prove the theorem it now suffices to describe the dynamic programming steps that are to be carried out when computing the records of a join, introduce or forget node $t$ from the records of its children.

\smallskip
\noindent \textbf{$t$ is a Join node with children $t_1$ and $t_2$.} \quad
We gradually construct $\rec(t)$ as follows. We initialize by having $\rec(t)$ map each configuration of $t$ to $\infty$.
We loop over all records of $t_1$ and $t_2$, and for each pair $(X_V^1,X_E^1,\theta^1)\in \rec(t_1)$, $(X_V^2,X_E^2,\theta^2)\in \rec(t_2)$. We check whether (1) $\theta^1=\theta^2$ and whether (2) $X_V^1\cap X_V^2=\emptyset$. If these checks succeed, we set $\rec(t)(X_V^1\cup X_V^2, X_E^1\cup X_E^2, \theta^1):=\min(\rec(t)(X_V^1\cup X_V^2, X_E^1\cup X_E^2, \theta^1), \quad \rec(t_1)(X_V^1, X_E^1, \theta^1)+\rec(t_2)(X_V^2, X_E^2, \theta^2))$.

\smallskip
\noindent \textbf{$t$ is a Forget node with child $t'$.} \quad
We once again gradually construct $\rec(t)$ and initialize by having $\rec(t)$ map each configuration of $t$ to $\infty$. Let $w= \chi(t')\setminus \chi(t)$ be the sector forgotten at $t$. We loop over all configurations at $t'$, and for each such configuration $(X_V',X_E',\theta')$ we construct a configuration at $t$ by (1) restricting $\theta'$ only to $\chi(t')\setminus \{w\}$, and (2) adding to $X_V'$ all vertices which $\theta'$ placed in $w$, and (3) adding to $X_E'$ all edges which were present in $\theta'(w)$ and which are not present in $\theta'(v)$ for any other $V\in \chi(t)$. We remark that since $\chi(t)$ is a separator, point (3) guarantees that any such edge added to $X_E'$ must have been fully processed in $\chi_\downarrow(t)\setminus \chi(t)$. Let $(X_V,X_E,\theta)$ be the configuration at $t$ constructed in this way. We then set $\rec(t)(X_V,X_E,\theta):=\min(\rec(t)(X_V,X_E,\theta), \rec(t')(X_V',X_E',\theta') + \upsilon)$ where $\upsilon$ is the number of bends occurring in $\theta'(w)$.

\smallskip
\noindent \textbf{$t$ is an Introduce node with child $t'$.} \quad
Let $w=\chi(t)\setminus \chi(t'$ by the sector introduced at $t$, and let us once again initialize by having $\rec(t)$ assign all configurations at $t$ the value $\infty$. To compute the records at $t$, we loop over all configurations at $t'$ and also over all of the at most $k^{\bigoh(\gridsize(k)^2)}$ possible choices of $\gridsols(w)$. For each such configuration $(X_V',X_E',\theta')$ at $t'$ and each $\phi\in\gridsols(w)$, we construct the configuration $(X_V',X_E',\theta'\cup \{w\mapsto \phi\})$ and simply set $\rec(t)(X_V',X_E',\theta'\cup \{w\mapsto \phi\})=\rec(t')(X_V',X_E',\theta')$.

The total running time required to process each node can be upper-bounded by the time required to process join nodes, which is at most $2^{\bigoh(k)}$ (for looping over all choices of $X_V^1,X_E^1,X_V^2,X_E^2$) times $(k^{\bigoh(k\cdot \gridsize(k)^2)})$ (for looping over all choices of $\theta^1, \theta^2$).

To argue correctness, let us consider an instance $\mathcal{I}$ of \FTBOE\ with solution $\beta$. This means that $\mathcal{I}$ admits a $\beta$-extension, and by Lemma~\ref{lem:sectorgrid} $\mathcal{I}$ also admits a $\beta$-extension which uses only the points of the sector-grid for drawing vertices and bends. Then at the root node $r$ of $T$, assuming the dynamic programming algorithm correctly computed the records we would obtain that $\rec(r)(V(G_f)\setminus V(H_f), E(G_f)\setminus E(H_f), \emptyset)=\beta$, as desired. Hence, it suffices to verify the correctness of the computation of the records for join, introduce, and forget nodes. In all three cases, this verification follows by recalling the definition of the records and observing how the records of $t$ depend on the records of its children.
\end{proof}

By combining Lemma~\ref{lem:dynprog} with Lemma~\ref{lem:reduction} and Observation~\ref{obs:sectorbound}, we conclude:
\begin{corollary}
\label{cor:main}
\TBOE can be solved in time $2^{\kappa^{\bigoh(1)}} \cdot n$, where $n$ is the number of feature points of $\Gamma(H)$.
\end{corollary}

\section{Concluding Remarks}
\label{sec:concl}
We have established the fixed-parameter tractability of the extension problem for bend-minimal orthogonal drawings, marking a notable addition to our understanding of drawing extension problems. What distinguishes this result from some of its predecessors on, e.g., extending 1-planar~\cite{eghkn-enc1dpt-20}, simple $k$-planar~\cite{GanianHKPV21} or crossing-minimal~\cite{HammH22} drawings, is that  these examples were topological while orthogonal planar drawings are geometric in nature. We believe this is one of the reasons why it seems impossible to use previously developed techniques in our setting, a fact which inspired the development of a novel machinery that we believe will find applications beyond the specific context of the problem studied here. 

As an example of this, a minor adjustment of our technique is already sufficient to also obtain a fixed-parameter algorithm for the problem of extending an orthogonal planar drawing while preserving a bound $\delta$ on the number of bends per edge~\cite{BlasiusRW16,bk-bhogd-98} parameterized by $\kappa+\delta$. But the technique could also possibly be applied to more general drawing styles, such as extending drawings restricted to boundedly many allowed edge slopes~\cite{DBLP:journals/jgaa/Hoffmann17,DBLP:journals/siamdm/KeszeghPP13}.

\bibliography{GD-ref.bib}

\end{document}